\documentclass[a4paper,11pt,fleqn,english,draft]{scrartcl}
\usepackage[utf8]{inputenc}
\usepackage[english]{babel}
\usepackage[T1]{fontenc}
\usepackage{enumerate}
\usepackage{dsfont}
\usepackage{mathrsfs}
\usepackage{amsmath,amsthm,amsfonts,amssymb}
\usepackage{amsbsy}
\usepackage{mathtools}

\usepackage{hyperref}
\hypersetup{
  bookmarksnumbered=true,
  bookmarksopen,
  pdfauthor=Martin\ Oelker,
  final=true
}

\addtokomafont{title}{\normalfont\bfseries}
\addtokomafont{author}{\normalfont}
\addtokomafont{date}{\itshape}
\addtokomafont{section}{\normalfont\bfseries}
\addtokomafont{subsection}{\normalfont\bfseries}
\addtokomafont{subsubsection}{\normalfont\bfseries}
\addtokomafont{paragraph}{\normalfont\bfseries}
\addtokomafont{sectionentry}{\normalfont}

\if@mathematic
\def\vec#1{\ensuremath{\mathchoice
		{\mbox{\boldmath$\displaystyle\mathbf{#1}$}}
		{\mbox{\boldmath$\textstyle\mathbf{#1}$}}
		{\mbox{\boldmath$\scriptstyle\mathbf{#1}$}}
		{\mbox{\boldmath$\scriptscriptstyle\mathbf{#1}$}}}}
\else
\def\vec#1{\ensuremath{\mathchoice
		{\mbox{\boldmath$\displaystyle#1$}}
		{\mbox{\boldmath$\textstyle#1$}}
		{\mbox{\boldmath$\scriptstyle#1$}}
		{\mbox{\boldmath$\scriptscriptstyle#1$}}}}
\fi

\renewcommand{\H}{\mathcal{H}} 							
\newcommand{\N}{\mathbb{N}} 								
\newcommand{\R}{\mathbb{R}} 								
\newcommand{\C}{\mathbb{C}} 								
\newcommand{\dd}{\,\mathrm{d}^3} 						
\newcommand{\D}{\mathcal{D}} 							
\newcommand{\norm}[1]{\left\|#1\right\|}					
\newcommand{\Abs}[1]{\left|#1\right|}						
\newcommand{\F}{\mathcal{F}}								
\DeclareMathOperator{\Ker}{Ker} 

\newcommand{\HDC}{H_{\mathrm{DC}}}

\newcommand{\Hrel}{\H^{\mathrm{rel}}}
\newcommand{\Hcom}{\H^{\mathrm{com}}}
\newcommand{\Mminp}{\mathbf{M}^- \cdot \hat{\vec{p}}}
\newcommand{\PMpl}{\hat{\vec{P}} \cdot \mathbf{M}^+}

\theoremstyle{plain}
\newtheorem{definition}{Definition}
\newtheorem{proposition}{Proposition}
\newtheorem{theorem}{Theorem}

\newtheorem{lemma}{Lemma}
\theoremstyle{definition}
\newtheorem{remark}{Remark}

\title{Distinguished self-adjoint extension\\ of the two-body Dirac operator\\ with Coulomb interaction%
\thanks{This work was partially funded by the Elite Network of Bavaria through the Junior Research Group “Interaction between Light and Matter”. M.O.\ gratefully acknowledges financial support from the German Academic Scholarship Foundation. }}

\author{Dirk-André Deckert
	\thanks{Math. Inst. Universität München, 80333 München, Germany, \href{mailto:deckert@math.lmu.de}{deckert@math.lmu.de}}
	\and 
	Martin Oelker
	\thanks{Math. Inst. Universität München, 80333 München, Germany, \href{mailto:oelker@math.lmu.de}{oelker@math.lmu.de}}
}

\date{\today}

\begin{document}
\maketitle

\begin{abstract}
\noindent\textbf{Abstract} We study the two-body Dirac operator in a bounded external field and
for a class of unbounded pair-interaction potentials, both repulsive and
attractive, including the Coulomb type. Provided the coupling constant of the
pair-interaction fulfills a certain bound, we prove existence of a self-adjoint
extension of this operator which is uniquely distinguished by means of finite
potential energy. In the case of Coulomb interaction, we require as a technical assumption the coupling constant to be bounded by $2/\pi$. \\

\end{abstract}

\section{Introduction}

Self-adjoint extensions of one-particle Dirac operators in an external field
have received much attention in the past and are still subject of very active research (see, e.g., \cite{wuest_dis-sa-ext},
\cite{nenciu_sa_dirac}, \cite{klaus_wuest_sa-ext}, \cite{esteban2007self},
\cite{arrizabalaga2013}, \cite{gallonemichel}, and \cite{gallone} for a great survey). As these operators are not bounded below, no canonical
extension in the sense of the Friedrichs extension for bounded below operators
is available.  Therefore, one main objective in the literature is
the determination of distinguished self-adjoint extensions and their
classification. 

For more than one particle, one usually considers the so-called
Brown-Raven\-hall operators, i.e., many-body Dirac operators projected to the
(carefully chosen) positive energy subspace. These operators are bounded below,
and thus, the Friedrichs extension can be exploited (see, e.g.,
\cite{tix1997self}, \cite{evansperrysiedentop_specreloneelecatoms},
\cite{morozovvugalter_stabilityatomsbrownravenhall}). In these settings,
mainly spectral questions are under investigation (see \cite{matte2010spectral} and references therein). 

In this article, we study a self-adjoint extension of the unprojected two-body
Dirac operator for unbounded potentials such as the Coulomb type.
Such operators are generically unbounded below and therefore often dismissed as
unphysical. Nonetheless, they are frequently employed in numerical studies in
relativistic quantum chemistry---see \cite{liu_cat} for an overview and
concerning spectral properties, e.g., \cite{jauregui_upperbounds},
\cite{liu_cat}, \cite{pestka_application}, \cite{pestka_complex},
\cite{watanabe2007effect}, \cite{watanabe2010effect}. It is therefore desirable
to understand these unprojected two-body Dirac operators better also from a mathematical point of view.  This was recently emphasized by
Dereziński in \cite{derezinski_openproblems}. Furthermore, the presented
technique may be a stepping stone to give insight into the domain properties of
the generator of the time evolution of the Dirac sea
\cite{dd_externalfield,deckert_external_2016}. 

To our best knowledge, there is only one publication so far that touches upon
this topic \cite{okaji_dere}. In their work, the two-body Dirac
operator for Coulomb interaction potential and Coulomb external potential is
studied.  Unfortunately, the given proof of essential self-adjointness
comprises a gap, which is briefly discussed in Appendix~\ref{section_exproof} below. 

Here, we prove existence of a self-adjoint extension for a class of two-body Dirac operators. Besides existence, it is also desirable to provide a criterion that distinguishes the extension uniquely and is physically meaningful at the same time. We adopt the criterion of finite potential energy. This criterion is well-known from the one-particle case (see
\cite{wuest_dis-sa-ext}). 

The main difficulty we face is of technical nature. The free two-body operator
exhibits a non-trivial nullspace in the coordinate of the interaction. This
nullspace is hidden in the standard representation of the Dirac matrices. Thus,
an unbounded interaction potential cannot be relatively bounded by the free
operator and a lot of standard perturbation techniques based on
such a bound are not applicable. Instead we use a Frobenius-Schur factorization
based on this nullspace and its orthogonal complement and infer
self-adjointness of the full two-body Dirac operator from self-adjointness of
the Schur complement.  The quadratic form techniques involved in this require the use of the theory of Calder\'{o}n-Zygmund singular
integrals. We want to remark that a similar method, although in a different
context, has also been used in \cite{loss_esteban_schur}. 

\section{Main results and strategy of proof}

The central object of study of this article is $\HDC$, the two-body Dirac operator with interaction potential $V_{\mathrm{int}}$ in the presence of an external potential $V_{\mathrm{ext}}$. The relevant Hilbert space is the twofold tensor product of the Hilbert space of $\C^{4}$-valued, square integrable functions, i.e.,
\begin{align}\label{eq_def_twoparticleHilbertspace}
\H_2 := L^2 (\R^3, \dd x) \otimes \C^4 \otimes L^2 (\R^3, \dd y) \otimes \C^4 .
\end{align} 
On $\H_2$, the symbolic expression of $\HDC$ takes the form
\begin{align}\label{eq_HDCdef}
\HDC 
:= 
H_0 + V_{\mathrm{ext}} + V_{\mathrm{int}},
\end{align}
where the free two-body Dirac operator $H_0$ is given by
\begin{align}\label{eq:freedirac}
H_0 := (- i\vec{\alpha} \cdot \nabla_{\vec{x}} + \beta m_1) \otimes \mathrm{id} 
+ \mathrm{id} \otimes (- i\vec{\alpha} \cdot \nabla_{\vec{y}} + \beta m_2) 
\end{align}
in the units in which both the speed of light $c$ and Planck's constant $\hbar$ equal one. Moreover, $m_1, m_2 \geq 0$ denote the masses of the two
particles and $\nabla_{\vec{x}}$, $\partial/\partial x_i$, etc.,
denote the gradient with respect to $\vec x=(x_1,x_2,x_3)^{\top}\in\R^3$ and
partial derivative with respect to $x_i$, $i=1,2,3$, respectively. We
denote by $\mathrm{id}_X$ the identity on the space $X$---however,
wherever unambiguous, we usually drop the subscript $X$. As it is helpful to
distinguish the identity on $L^2$ from the identity matrix on $\C^{n}$, we
denote the latter by $\mathbf{1}_n$. 

We assume $V_{\mathrm{ext}}$ to be bounded and symmetric, i.e., for all $f,g
\in \H_2$ we have 
\begin{align}\label{eq_definitionV_ext}
\left|\left\langle 
g, V_{\mathrm{ext}} f
\right\rangle\right|
=
\left|\left\langle
V_{\mathrm{ext}} g, f
\right\rangle\right|
< \infty .
\end{align}
The precise form of $V_{\mathrm{ext}}$ plays no role concerning
self-adjointness, we have, however, a finite nucleus model of Coulomb type with
regularized singularity in mind. 

As regards the interaction potential, we define $V_{\mathrm{int}}$  for almost
all $\vec{x}, \vec{y} \in \R^3$ as multiplication operator, i.e.,  
\begin{align}\label{eq_definitionV_int}
\left( V_{\mathrm{int}} f \right) (\vec{x}, \vec{y}) 
:=
\mathbf{1}_{16} \, \frac{\gamma}{| \vec{x} - \vec{y} |^\kappa} \, f (\vec{x}, \vec{y}) ,
\end{align}
where $\gamma \in \R$ is the coupling constant, $0< \kappa \leq 1$ controls the strength of the singularity, and where $f \in \H_2$ fulfills $V_{\mathrm{int}} f \in \H_2 $. Usually, we drop the identity matrix $\mathbf{1}_{16}$. We note that Coulomb interaction
potentials, i.e., $\kappa=1$, are included in class~\eqref{eq_definitionV_int}. 

The canonical domain of $H_0$ as well as of $\HDC$ is 
\begin{align}\label{eq:D0}
\D_0 := H^1 (\R^3, \dd x) \otimes \C^4 \otimes H^1 (\R^3, \dd y) \otimes \C^4,
\end{align}
where $H^k$ denotes the $k$-th Sobolev space over $L^2$. 

Furthermore, the Hermitian matrices $\vec{\alpha} = (\alpha_1, \alpha_2, \alpha_3)$ and $\beta$ are the so-called Dirac matrices in standard
representation 
\begin{align}\label{eq_diracmatrices}
\vec{\alpha} = 
\left(
\begin{array}{cc}
0 & \vec{\sigma} \\ 
\vec{\sigma} & 0
\end{array}
\right) 
\, , \quad 
\beta = 
\left( 
\begin{array}{cc}
\mathbf{1}_2 & 0 \\
0 & -\mathbf{1}_2
\end{array}
\right) , 
\end{align} 
where $\vec{\sigma}=(\sigma_1, \sigma_2, \sigma_3)$ are the Pauli matrices 
\begin{align}
\sigma_1 = 
\left( 
\begin{array}{cc}
0 & \,1 \\ 
1 & \,0
\end{array}
\right)
\, , \quad 
\sigma_2 = 
\left( 
\begin{array}{cc}
0 & -i \\ 
i & 0
\end{array}
\right)
\, , \quad 
\sigma_3 = 
\left( 
\begin{array}{cc}
1 & 0 \\ 
0 & -1
\end{array}
\right). 
\end{align}
The Dirac matrices obey the following anticommutation relations 
\begin{align} \label{clifford_1}
&\alpha_k \alpha_l + \alpha_l\alpha_k = 2 \delta_{kl} \mathbf{1}_4\;,\quad k, l = 1, 2, 3, \nonumber \\
&\alpha_i\beta + \beta\alpha_i = 0 \;,\quad i=1,2,3, \\
&\beta^2 = \mathbf{1}_4 , \nonumber
\end{align}
where $\delta_{kl} = 1$ if $k = l$, and zero else. \\

Our results about self-adjointness are collected in the following Theorem~\ref{thm:main1}. As this paper is structured roughly according to its claims \ref{thm:main1_a})--\ref{thm:main1_c}), we give the proof of each of these claims in the corresponding section. See the article outline below. 

\begin{theorem}\label{thm:main1}
Let $0 < \kappa \leq 1$ and $|\gamma| M_{\kappa/2}^2 < 1$, where $M_{\kappa} > 0$ is given by
\begin{align}\label{eq_def_Mkappa}
M_{\kappa} 
:= 
2^{- \kappa}  \frac{\Gamma \left( \frac{3}{4} - \frac{\kappa}{2} \right)}{\Gamma \left( \frac{3}{4} + \frac{\kappa}{2} \right)} . 
\end{align}
\begin{enumerate}[a)]
	\item \label{thm:main1_a}
	$H_0$ is self-adjoint on its natural domain $\D (H_0) = \{ f \in \H_2 |\ H_0 f \in \H_2 \}$.
	
	\item \label{thm:main1_b}
	There exists a self-adjoint extension of $\HDC$, denoted by $\tilde{H}_{\mathrm{DC}}$, with domain $\D(\tilde{H}_{\mathrm{DC}})$, given in Eq.~\eqref{eq_def_domain_Htilde_F}. 
	
	\item \label{thm:main1_c}
	$\tilde{H}_{\mathrm{DC}}$ is the unique self-adjoint extension that fulfills for all 
	$f \in \D (\tilde{H}_{\mathrm{DC}})$ the condition 
	\begin{align}
	\left| E_{\mathrm{pot}} [f] \right| := \left| \left\langle f, \left( V_{\mathrm{ext}} + V_{\mathrm{int}} \right) f \right\rangle \right| < \infty
	\end{align}
	and that can be split into a relative and a center-of-mass part, as it is made precise in Theorem~\ref{thm_distinguishedextension}, in particular, Eq.~\eqref{eq_def_Htilde}. 
\end{enumerate}
\end{theorem}

\begin{remark}\label{rem_aftermainthm} 
We want to remark:
\begin{enumerate}[a)]
	\item Although part a) of Theorem~\ref{thm:main1} comes as no surprise, it
	turns out that, contrary to what one might have expected, the domain of
	self-adjointness $\D (H_0)$ is {\em not} contained in the tensor product of
	the form domains of two one-particle Dirac Hamiltonians. Heuristically, the reason for this phenomenon is that two particles described by $H_0$ can both have infinite kinetic energy as long as these energies cancel. This will be made clear in Remark~\ref{rem_unsualstructure}. 
	
	\item For $\kappa = 1$, i.e., the case of Coulomb interaction, the condition on the coupling constant is $|\gamma| < 2/\pi$. The smaller $\kappa$ is chosen, the larger values of $|\gamma|$ are allowed. 
	
	\item We have no reason to believe that the restriction on the coupling constant $\gamma$ is optimal in the sense that larger values of $|\gamma |$ would not allow for self-adjoint extensions anymore. We believe it is mainly due to our method of proof. Techniques from the theory of self-adjoint extensions of the one-particle Dirac operator with external potential may be applied in order to obtain larger values of $|\gamma |$. This is however not our focus here. 
\end{enumerate}
\end{remark}

\paragraph{Article outline and strategy of proof}

After a change of coordinates to relative and center-of-mass coordinates, introduced in \eqref{eq_transformationtorelcomcoord} of Section~\ref{section_coordtransform}, the two-body Dirac-Coulomb Hamiltonian $\HDC$ splits into a relative and a center-of-mass Hamiltonian; cf.\ \eqref{eq_HDCdef_2}, \eqref{eq_def_domain_comHamiltonian}, \eqref{eq_defHrel_2} in Section~\ref{section_relevantdefs}. The coefficient matrices of the total momentum operator and the relative momentum operator, the matrices $\mathbf{M}^+$ and $\mathbf{M}^-$ defined in \eqref{eq_def_Mpl_Mmin} of Section~\ref{section_projections}, no longer obey the anticommutation relations \eqref{clifford_1}. They exhibit a non-trivial nullspace structure which carries over to the free relative and center-of-mass Hamiltonians and does not permit a straightforward application of Kato-Rellich perturbation theory to probe self-adjointness (see Lemma~\ref{lemma_Vnotrelbounded}). For their study, we are led to introduce the projection $P_-$ which projects on the nullspace of the free relative Hamiltonian, as well as the projection on the orthogonal complement, $P_+$. In Section~\ref{section_definitionpreparations}, we introduce all these basic objects and collect their fundamental properties relevant to our study. 

Section~\ref{section_freetwobody} provides all the needed properties of the free two-body Dirac Hamiltonian, in particular, the proof of claim \ref{thm:main1_a}) of our main result Theorem~\ref{thm:main1}. 

In order to study $\HDC$, the orthogonal projections $P_+$ and $P_-$ are employed to split the Hilbert space of the relative coordinate into two orthogonal subspaces, i.e., $\Hrel = \Hrel_+ \oplus \Hrel_-$. This splitting is not left invariant by the Coulomb interaction. Therefore, one is naturally led to consider the parts of the relative Hamiltonian $H^{\mathrm{rel}}$ in the various subspaces, i.e., we obtain a $2 \times 2$-matrix representation of $H^{\mathrm{rel}}$ whose entries are unbounded operators and given by $P_\pm H^{\mathrm{rel}} P_\pm$ on the diagonal and $P_\pm H^{\mathrm{rel}} P_\mp$ on the off-diagonal, respectively; see \eqref{eq_defHrel} in Section~\ref{section_relevantdefs}. 

In order to construct a self-adjoint extension of $\HDC$, which is the content of Section~\ref{section_self-adjointextension}, we will use that---under some conditions---self-adjointness of matrix operators is encoded in the so-called Schur complement, denoted by $S$: First, we will construct a self-adjoint extension of $S$, namely $S_F$, cf.\ Lemma~\ref{lemma_schursym} of Section~\ref{section_self-adjointextension}. This will then allow us to define a self-adjoint extension $H^{\mathrm{rel}}_F$ of $H^{\mathrm{rel}}$, see Theorem~\ref{thm_self-adjointextensionHrel} of Section~\ref{section_self-adjointextension}, which finally paves the way to the self-adjoint extension of $\HDC$, denoted by $\tilde{H}_{\mathrm{DC}}$, in Theorem~\ref{thm:sa} of Section~\ref{section_self-adjointextension}. This proves claim~\ref{thm:main1_b}) of Theorem~\ref{thm:main1}. We want to remark that a similar strategy, although in a different setting, has been employed in \cite{loss_esteban_schur}. 

Furthermore, we prove that the interaction potential is not relatively bounded by the free relative Hamiltonian (Lemma~\ref{lemma_Vnotrelbounded}). 

As $S_F$ is given as a form sum, it will turn out to be very convenient to compute the closure of the form associated with $S$ explicitly. At the heart of this computation lies Theorem~\ref{thm_coreV1/2} whose proof is quite lengthy and therefore given in the separate Section~\ref{section_prooftechnicalstuff}. It uses the Calderón-Zygmund theory of singular integrals. 

That $\tilde{H}_{\mathrm{DC}}$ is a distinguished self-adjoint extension of
$\HDC$, which is shown in the last Section~\ref{section_criterionextension},
holds in the following sense: Let $\tilde{H}$ be any self-adjoint extension of
$\HDC$ that can be split into a relative and a center-of-mass part. Then, $f
\in \D (\tilde{H})$ has finite potential energy if and only if $\tilde{H} =
\tilde{H}_{\mathrm{DC}}$. This is satisfying in two respects. First, it is a physically sensible criterion, and second, the criterion 
singles out $\tilde{H}_{\mathrm{DC}}$ uniquely. This proves
claim~\ref{thm:main1_c}) of Theorem~\ref{thm:main1}. 

The Appendix~\ref{section_appendix_okaji} contains a comment on \cite{okaji_dere}, the only publication of which we are aware that also treats self-adjointness of $\HDC$. 

In Appendix~\ref{section_matrixoperators}, all needed tools to study matrix operators with unbounded entries are collected. Most notably, we introduce the Frobenius-Schur factorization of a matrix operator, see \eqref{eq_generalmatrixoperator} and Theorem~\ref{thm_closability}.

\section{Proofs}

\subsection{Definitions and preliminary results}\label{section_definitionpreparations}

\subsubsection{Coordinate transformation, Hilbert spaces, Fourier transform} \label{section_coordtransform}
It will be convenient to introduce a change of the coordinates $\vec{x}, \vec{y} \in \R^3$ by means of 
\begin{align}
\vec{r} := \vec{x}-\vec{y}\, , \quad  \vec{R} := \frac{1}{2} (\vec{x} + \vec{y}). 
\end{align}
Here, $\vec r \in \R^3$ is the relative coordinate (abbreviated by rel) and $\vec{R} \in \R^3$ the center-of-mass coordinate (com). 
Furthermore, we define with 
$\vec{X} := (\vec{x}, \vec{y})$ 
and 
$\vec{Y} := (\vec{r}, \vec{R})$ 
the transformation matrix 
$\mathrm{U} \colon \R^3 \times \R^3 \rightarrow \R^3 \times \R^3$ 
as 
\begin{align} 
\mathrm{U}\vec{X} := 
\left(
\begin{array}{cc}
1 & - 1 \\
\frac{1}{2} & \frac{1}{2} 
\end{array}
\right) 
\left(
\begin{array}{c}
\vec{x} \\
\vec{y}
\end{array}
\right) 
= 
\left(
\begin{array}{c}
\vec{r} \\
\vec{R}
\end{array}
\right) 
=
\vec{Y} .
\end{align}
Since $\det \mathrm{U} = 1$, it induces a unitary transformation on the two-particle Hilbert space 
\begin{align}\label{eq_transformationtorelcomcoord}
U 
\colon 
\H_2
\rightarrow 
L^2 (\R^3, \dd R) \otimes \C^{16} \otimes L^2 (\R^3, \dd r)
\end{align}
given by 
$(U f) (\vec{Y}) := f (\vec{X}) = f (\mathrm{U}^{-1} \vec{Y})$. We define the Hilbert spaces 
\begin{align}\label{eq_def_Hrel_Hcom}
\Hrel := \C^{16} \otimes L^2 (\R^3, \dd r)
\, , \quad
\Hcom := L^2 (\R^3, \dd R) \otimes \C^{16} .
\end{align}
Norm and scalar product on all of the used Hilbert spaces are denoted by $\| \cdot \|$ and $\langle \cdot , \cdot \rangle $, respectively. In some cases, possible confusion is avoided by suitable subscripts. The scalar product on $\Hrel$ of $f,g \in \Hrel$ is defined as 
\begin{align}
\langle f, g \rangle
:=
\int_{\R^3} f^\dag (\vec{r}) \, g (\vec{r}) \dd r 
= 
\int_{\R^3} \sum_{k=1}^{16} \overline{f^k (\vec{r})} \, g^k (\vec{r}) \dd r 
\end{align}
and in analogy to that in the other Hilbert spaces. $f^k$ is the $k$-th component of the $\C^{16}$-spinor $f$. $\overline{z}$ denotes complex conjugation of $z \in \C$. Instead of $f (\vec{r})^\dag f (\vec{r})$, we will often just write $\Abs{f (\vec{r})}^2$. When finite, $\| \cdot \|$ also denotes the norm of a linear operator. The context will always distinguish it from the $L^p$-norm. The operator closure of an arbitrary, but closable linear operator $A$ is denoted by $\overline{A}$. No confusion with complex conjugation will arise. 

We define, as it is usually done for square-integrable, $\C$-valued functions, the following Fourier transform
on $L^2 (\R^3)$ for almost all $\vec{p} \in \R^3$ by 
\begin{align}
\hat{f} (\vec{p}) :=
(\F f) (\vec{p}) := 
\lim_{M \rightarrow \infty} \int_{|\vec{r}| \leq M} e^{- 2 \pi i\vec{r} \cdot \vec{p}} f (\vec{r}) \dd r
\end{align}
where the limit is taken in the $L^2$-sense. In some cases, the notation $\F_{\vec{R}}$ and $\F_{\vec{r}}$ clarifies, whether the Fourier transform is taken with respect to the center-of-mass coordinate or the relative coordinate. This definition carries over to $L^2 (\R^3) \otimes \C^n$ by applying the transformation component-wise. We introduce the notation for the relative momentum operator $\hat{\vec{p}}= - i \nabla_{\vec{r}}$ and the total momentum operator $\hat{\vec{P}}= - i \nabla_{\vec{R}}$ and remark with respect to this notation that it denotes both, the differential operators $- i \nabla_{\vec{R}}$ and $- i \nabla_{\vec{r}}$, respectively, when acting on $f$ as well as multiplication with $\vec{P} \in \R^3$ and $\vec{p} \in \R^3$, respectively, when acting on $\hat{f}$ in Fourier space. With $p^2$ we mean $\Abs{\vec{p}}^2$ for any $\vec{p} \in \R^3$. With $\hat{p}^2$, however, we denote the operator product $\hat{\vec{p}}^2$. Moreover, operators that are composed of $\hat{\vec{P}}$ or $\hat{\vec{p}}$ are defined with help of the Fourier transform. E.g., the operator $(\hat{p}^2 + 1)^{\kappa/2}$ that will appear in Section~\ref{section_self-adjointextension} is defined in Fourier space as multiplication by $(p^2 + 1)^{\kappa/2}$ with $\vec{p} \in \R^3$. 

\subsubsection{The projections $P_+$ and $P_-$} \label{section_projections}
In this section, the projections $P_+$ and $P_-$ are introduced. They are central to the study of $\HDC$ insofar as they reveal the technical obstacle that does not allow to use standard techniques to study $\HDC$. This is seen explicitly in Lemma~\ref{lemma_Vnotrelbounded} which opens Section~\ref{section_self-adjointextension}. \\

The reader should be warned that $P_+$ and $P_-$ are \textit{not} the spectral projections of $H_0$, i.e., the operators that project on the positive and negative part of the spectrum of $H_0$. What $P_+$ and $P_-$ project onto, is the content of Proposition~\ref{prop_justificationPplusmin}. \\

We define the $16 \times 16$-matrices 
$\mathbf{M}^{\pm} = \left( \mathrm{M}^{\pm}_1, \mathrm{M}^{\pm}_2, \mathrm{M}^{\pm}_3 \right)$
by 
\begin{align}\label{eq_def_Mpl_Mmin}
\mathbf{M}^+ := \frac12 (\vec{\alpha} \otimes \mathbf{1}_4 + \mathbf{1}_4 \otimes \vec{\alpha})
\, , \quad 
\mathbf{M}^{-} := \vec{\alpha} \otimes \mathbf{1}_4 - \mathbf{1}_4 \otimes \vec{\alpha} .
\end{align}
They become relevant later on as coefficient matrices of $\hat{\vec{P}}$ and $\hat{\vec{p}}$, when transforming $\HDC$ to relative and center-of-mass coordinates. Consequently, we define $\Mminp$ in the underlying Hilbert space $\Hrel$ with domain 
\begin{align}\label{eq_def_domain_Mminp}
\D (\Mminp) = 
\left\{ 
f \in \Hrel \left|\ \Mminp f \in \Hrel \right. 
\right\} 
\end{align}
and $\PMpl$ in the underlying Hilbert space $\Hcom$ with domain 
\begin{align}\label{eq_def_domain_PMpl}
\D (\PMpl) = 
\left\{ 
f \in \Hcom \left|\ \PMpl f \in \Hcom \right. 
\right\} .
\end{align}

Without hat, $\mathbf{M}^- \cdot \vec{p}$ means the $16 \times 16$-matrix, whereas with hat, $\Mminp$ denotes an (unbounded) operator. 

\begin{proposition}\label{prop_kernmrelmcom}
${}$
\begin{enumerate}[a)]
	\item For all $\vec{p} \in \R^3$, we have $\dim \Ker (\mathbf{M}^- \cdot \vec{p} ) = 8$. 
	\item $\Ker (\Mminp )$ is isomorphic to $\C^8 \otimes L^2 (\R^3, \dd r)$. 
	\item $\Ker (\PMpl )$ is isomorphic to $L^2 (\R^3, \dd R) \otimes \C^8$. 
\end{enumerate}
\end{proposition}
\begin{proof}
\begin{enumerate}[a)]
	\item This follows as for all $\vec{p} \in \R^3$ one finds that $\mathbf{M}^- \cdot \vec{p}$ has the eigenvalue $0$ with multiplicity $8$. 
	
	\item As $\mathbf{M}^- \cdot \vec{p}$ is a Hermitian matrix, there exists a unitary matrix $u (\vec{p})$ which diagonalizes $\mathbf{M}^- \cdot \vec{p}$. We find for almost all $\vec{p} \in \R^3$ 
	\begin{align}\label{eq_diagmatrixmrel}
	u (\vec{p}) \, \mathbf{M}^- \cdot \vec{p} \, u (\vec{p})^\dag
	=
	2
	\left(
	\begin{array}{ccc}
	-\mathbf{1}_4 | \vec{p} | &  & \\
	& \mathbf{1}_4 | \vec{p} | & \\
	&  & \mathbf{0}_8 
	\end{array}
	\right) ,
	\end{align}
	where $u (\vec{p})^\dag$ denotes the adjoint of $u (\vec{p})$. In $L^2$, we have $\Ker (|\hat{\vec{p}}|) = \{ 0 \}$, and therefore, $\Ker (\Mminp)$ is determined solely by $\mathbf{0}_8$ in the lower right corner, i.e., the eigenvalue $0$ of $\mathbf{M}^- \cdot \vec{p}$ with multiplicity $8$. These eigenvalues in turn correspond to $8$ linearly independent eigenvectors of the matrix $\mathbf{M}^- \cdot \vec{p}$. Thus, $u (\hat{\vec{p}})$ establishes an isomorphism between $\Ker (\Mminp )$ and $\C^8 \otimes L^2 (\R^3, \dd r)$, which proves the statement. 
	
	\item Analogously to b). \qedhere
\end{enumerate}
\end{proof}

In order to give the definition of $P_\pm$, we define for almost all $\vec{p} \in \R^3$ the Hermitian $16 \times 16$-matrix 
\begin{align}\label{eq_tauofp}
\tau (\vec{p}) := - \frac{\vec{\alpha} \cdot \vec{p} \otimes \vec{\alpha} \cdot \vec{p}}{p^2} .
\end{align}
With help of the anticommutation relations \eqref{clifford_1} for the Dirac matrices, we obtain $\tau (\vec{p})^2 = \mathbf{1}_{16}$. This implies that multiplication with $\tau (\vec{p})$ defines a bounded operator on all of $\Hrel$.

\begin{definition}\label{def_tau_projections}
We define the operator 
$\tau \colon \Hrel \rightarrow \Hrel$ 
by its action on all $f \in \Hrel$ and for almost all $\vec{r} \in \R^3$ 
\begin{align}
(\tau f) (\vec{r}) 
:=
\lim_{M \rightarrow \infty} \int_{| \vec{p} | \leq M} e^{2 \pi i\vec{r} \cdot \vec{p}} \, \tau (\vec{p}) \hat{f} (\vec{p}) \dd p 
\end{align}
where the limit is taken in the $L^2$-sense. We define the operators 
$P_{\pm} \colon \Hrel \rightarrow \Hrel_{\pm}$ by
\begin{align} 
P_{\pm} := 
\frac{1}{2} 
\left( 
\mathrm{id} \pm \tau
\right) 
\end{align}
where $\mathcal{H}_{\pm}^{\mathrm{rel}} := P_{\pm} \mathcal{H}^{\mathrm{rel}}$. We also define for almost all $\vec{p} \in \R^3$ the $16\times 16$-matrix 
\begin{align}
P_\pm (\vec{p}) := \frac{1}{2} \left( \mathbf{1}_{16} \pm \tau (\vec{p}) \right) .
\end{align}
\end{definition}

For the moment, it suffices to define $P_\pm$ as Fourier multiplier. In Section \ref{section_prooftechnicalstuff} however, integral kernels are derived. 

\begin{proposition}\label{prop_justificationPplusmin} 
The following statements hold: 
\begin{enumerate}[a)]
	
	\item \label{prop_justificationPplusmin_a}
	$P_-$ is the orthogonal projection onto $\Ker (\Mminp)$, i.e., $\Hrel_- = \Ker (\Mminp)$. 
	
	\item \label{prop_justificationPplusmin_b}
	$P_+$ is the orthogonal projection onto $\Ker (\Mminp)^{\bot}$, i.e., $\Hrel_+ = \Ker (\Mminp)^{\bot}$. 
	
	\item \label{prop_justificationPplusmin_c}
	$P_+ \, \Mminp \, P_+ f = \Mminp f$ and $P_- \, \Mminp f = 0$ for all $f \in \D (\Mminp)$. 
\end{enumerate}
Furthermore, 
$\mathcal{H}^{\mathrm{rel}} =\mathcal{H}^{\mathrm{rel}}_+ \oplus \mathcal{H}^{\mathrm{rel}}_-$, 
and $\mathcal{H}^{\mathrm{rel}}_{\pm}$ are themselves Hilbert spaces. 
\end{proposition}

\begin{proof}
\begin{enumerate}[a)]
	
	\item First, we prove that $P_-$ is an orthogonal projection, i.e., $P_-^2 = P_-$, $P_-$ is bounded, and $P_-^* = P_-$. Since $(\vec{\alpha} \cdot \vec{p})^2 =\mathbf{1}_4 p^2$, we obtain $P_-(\vec{p})^2 = P_-(\vec{p})$ for almost all $\vec{p} \in \R^3$. Thus, $P_-^2 = P_-$ follows. 
	
	$P_-$ is bounded with $\| P_- \| = 1$ since $P_-(\vec{p})$ is a Hermitian matrix with the only eigenvalues being $0$ and $1$. This also implies self-adjointness $P_-^{\ast} = P_-$. 
	
	In order to show that $P_-$ projects onto $\Ker (\Mminp)$, we prove $\Hrel_- = \Ker (\Mminp)$. A computation shows for almost all $\vec{p} \in \R^3$ $\mathbf{M}^- \cdot \vec{p} P_- (\vec{p}) = 0$ and therefore, we obtain $P_- f \in \Ker (\Mminp)$ for all $f \in \Hrel$, i.e., $\Hrel_- \subseteq \Ker (\Mminp)$. 
	
	For the reverse inclusion $\Ker (\Mminp) \subseteq \Hrel_-$, we pick an $f \in \Ker (\Mminp)$ and show that $P_- f = f$. We use $u (\vec{p})$ from the proof of Proposition~\ref{prop_kernmrelmcom}, line~\eqref{eq_diagmatrixmrel}, and compute for almost all $\vec{p} \in \R^3$ 
	\begin{align}\label{eq_diagmatrixPmin}
	u (\vec{p}) P_- (\vec{p}) u (\vec{p})^\dag
	=
	\left(
	\begin{array}{cc}
	\mathbf{0}_8 & \\
	& \mathbf{1}_8 \\ 
	\end{array}
	\right) .
	\end{align}
	Recalling that in $L^2$ one has $\Ker (|\hat{\vec{p}}|) = \{ 0 \}$, we see that $\Mminp f = 0$ implies that the only non-zero components of $u (\hat{\vec{p}}) f$ are those components on which $u (\hat{\vec{p}}) P_- u (\hat{\vec{p}})^*$ acts as identity. Therefore, we can conclude that $P_- f = f$. 
	
	\item $P_+$ is an orthogonal projection by the same argument as for $P_-$. In order to show that $P_+$ projects on $\Ker (\Mminp)^{\bot}$, we compute for almost all $\vec{p} \in \R^3$ $P_- (\vec{p}) P_+ (\vec{p}) = P_+ (\vec{p}) P_- (\vec{p}) = 0$ and therefore, for all $f,g \in \Hrel$ 
	\begin{align}
	\langle P_+ f, P_- g \rangle = \langle f, P_+ P_- g \rangle = 0 .
	\end{align}
	
	\item The statement follows since the relations 
	\begin{align}
	P_+ (\vec{p}) \mathbf{M}^- \cdot \vec{p} P_+ (\vec{p}) &= \mathbf{M}^- \cdot \vec{p} \\
	P_- (\vec{p}) \mathbf{M}^- \cdot \vec{p} &= 0
	\end{align} 
	hold for almost all $\vec{p} \in \R^3$. 
\end{enumerate}
$\quad$\\
Since the projections $P_{\pm}$ are closed, $\mathcal{H}^{\mathrm{rel}}_{\pm}$ are closed subspaces of $\mathcal{H}^{\mathrm{rel}}$. With the inherited inner product from $\mathcal{H}^{\mathrm{rel}}$, it follows that they are themselves Hilbert spaces. This implies $\mathcal{H}^{\mathrm{rel}} =\mathcal{H}^{\mathrm{rel}}_+ \oplus \mathcal{H}^{\mathrm{rel}}_-$. 
\end{proof}

\subsubsection{Operators and domains}\label{section_relevantdefs} 
In line~\eqref{eq:freedirac}, we already defined the free two-particle Dirac operator with masses $m_1, m_2 \geq 0$ as 
\begin{align}
H_0 
= 
(- i\vec{\alpha} \cdot \nabla_{\vec{x}} + \beta m_1) \otimes \mathrm{id} 
+ \mathrm{id} \otimes (- i\vec{\alpha} \cdot \nabla_{\vec{y}} + \beta m_2) 
\end{align}
with domain
\begin{align}
\D_0 = H^1 (\R^3, \dd x) \otimes \C^4 \otimes H^1 (\R^3, \dd y) \otimes \C^4.
\end{align}
When external as well as interaction potential are included, it yields the Dirac-Coulomb Hamiltonian from line~\eqref{eq_HDCdef} which reads 
\begin{align}\label{eq_HDCdef_2}
\HDC 
= 
H_0 + V_{\mathrm{ext}} + V_{\mathrm{int}}
\end{align}
whose domain is also $\D_0$. 

Now, we apply the coordinate transformation $U$ to relative and center-of-mass coordinates \eqref{eq_transformationtorelcomcoord} to $H_0$ as well as to $\HDC$. It is here where the matrices  
$\mathbf{M}^{\pm} = \left( \mathrm{M}^{\pm}_1, \mathrm{M}^{\pm}_2, \mathrm{M}^{\pm}_3 \right)$ from line~\eqref{eq_def_Mpl_Mmin} enter our considerations. From lines \eqref{eq_def_Mpl_Mmin}--\eqref{eq_def_domain_PMpl}, we already know the free relative Hamiltonian $\Mminp$ with domain 
\begin{align}
\D (\Mminp) = 
\left\{ 
f \in \Hrel \left|\ \Mminp f \in \Hrel \right. 
\right\} 
\end{align}
in the underlying Hilbert space $\Hrel$ as well as the free center-of-mass Hamiltonian $\PMpl$ with domain 
\begin{align}\label{eq_def_domain_comHamiltonian}
\D (\PMpl) = 
\left\{ 
f \in \Hcom \left|\ \PMpl f \in \Hcom \right. 
\right\}
\end{align}
in the underlying Hilbert space $\Hcom$. Since the mass terms $\beta m_1$ and $\beta m_2$ do not play a role concerning self-adjointness, we set $m_1 = 0 = m_2$ for now. We obtain 
\begin{align}\label{eq_def_T}
T
:= {}&
U
\Big(
(- i\vec{\alpha} \cdot \nabla_{\vec{x}} ) \otimes \mathrm{id}_{L^2 (\dd x)\otimes \C^4} + \mathrm{id}_{L^2 (\dd y)\otimes \C^4} \otimes (- i\vec{\alpha} \cdot \nabla_{\vec{y}} ) 
\Big)
U^{-1} \nonumber \\
= {}&
\PMpl \otimes \mathrm{id}_{L^2 (\dd r)} + \mathrm{id}_{L^2 (\dd R)} \otimes \Mminp .
\end{align}
Note the different $L^2$-identities: In the upper line, the identities have a spin part, whereas in the lower line, they do not. 

For our convenience, we do not introduce a new symbol for the transformed domain and obtain therefore 
\begin{align}
\D_0
& = 
U 
\Big(
H^1 (\R^3, \dd x) \otimes \C^4 \otimes H^1 (\R^3, \dd y) \otimes \C^4
\Big) \nonumber \\
& = 
H^1 (\R^3, \dd R) \otimes \C^{16} \otimes H^1 (\R^3, \dd r) .
\end{align}

Since also $V_{\mathrm{ext}}$ plays no role concerning self-adjointness as it is bounded and symmetric (see \eqref{eq_definitionV_ext}), we set it to zero. It can be restored later on by means of a bounded perturbation. Hence, conjugating $\HDC$ with $U$ with masses $m_1 = 0 = m_2$ and $V_{\mathrm{ext}} = 0$ yields 
\begin{eqnarray}
U \HDC U^{-1}
& = &
T + \mathrm{id}_{L^2 (\dd R)} \otimes \gamma V \nonumber \\
& = &
\PMpl \otimes \mathrm{id} + \mathrm{id} \otimes \big( \Mminp + \gamma V \big) . 
\end{eqnarray}
where, in the underlying Hilbert space $\Hrel$, $V$ is the operator of component-wise multiplication with $| \vec{r} |^{-\kappa}$ for almost all $\vec{r} \in \R^3$ and all $0 < \kappa \leq 1$ 
with domain 
\begin{align}
\D (V) = 
\left\{
f \in \Hrel \left|\ |\cdot |^{-\kappa} f \in \Hrel \right.
\right\} .
\end{align} 
It is well-known that $V$ is positive and self-adjoint on $\D (V)$. 

We saw in Section~\ref{section_projections}, Prop.~\ref{prop_justificationPplusmin}, that $P_+$ and $P_-$ split $\Hrel$ into two orthogonal subspaces, i.e., $\Hrel = P_+ \Hrel \oplus P_- \Hrel = \Hrel_+ \oplus \Hrel_-$. This means that we can recast 
\begin{align}\label{eq_defHrel_2}
H^{\mathrm{rel}} := \Mminp + \gamma V
\end{align}
in matrix form as  
\begin{eqnarray}\label{eq_defHrel}
H^{\mathrm{rel}} 
& = & 
P_+ H^{\mathrm{rel}} P_+ + P_+ H^{\mathrm{rel}} P_- + P_- H^{\mathrm{rel}} P_+ + P_- H^{\mathrm{rel}} P_- \nonumber \\
& =: & 
\left(
\begin{array}{cc}
\Mminp + P_+ \gamma V P_+ & \quad P_+ \gamma V P_-\\
P_- \gamma V P_+ & \quad P_- \gamma V P_-
\end{array}
\right) 
\end{eqnarray} 
where we already applied $P_+ \Mminp P_+ = \Mminp$, proven in Proposition~\ref{prop_justificationPplusmin}\ref{prop_justificationPplusmin_c}). The domain of $H^{\mathrm{rel}}$ is 
\begin{align} \label{eq_def_domain_Hrel}
\D (H^{\mathrm{rel}}) = 
\D_+ 
\oplus 
\left( 
\mathcal{D} (V) \cap \Hrel_-
\right)
\end{align}
where 
\begin{align}\label{eq_def_Dplus}
\D_+ := \left( \C^{16} \otimes H^1 (\R^3, \dd r) \right) \cap \Hrel_+ . 
\end{align}
For sake of completeness, we show that $\D (H^{\mathrm{rel}})$ is dense in $\Hrel $. 

\begin{proposition}\label{prop_domainH^reldense}
$\D (H^{\mathrm{rel}})$ is dense in $\Hrel$. 
\end{proposition}
\begin{proof}
The statement of the proposition follows if $\D_+$ is dense in $\Hrel_+$ and
$\mathcal{D} (V) \cap \Hrel_-$ is dense in $\Hrel_-$ since the splitting of $\Hrel$ is orthogonal. The former statement follows from density of $H^1$ in $L^2$, the latter holds then by Hardy's inequality. 
\end{proof}

\subsection{The free two-body Dirac operator}\label{section_freetwobody}
Before we include the interaction, it will be essential to study the free two-body Dirac operator given in \eqref{eq:freedirac}. In this section, we will provide the proof of claim a) of Theorem~\ref{thm:main1}, and furthermore, Lemma~\ref{lem_unsualstructure} that will be helpful in the next section when the interaction is included. 

\begin{theorem}[Claim~\ref{thm:main1_a}) of Theorem~\ref{thm:main1}]\label{thm_selfadjointnessfreetwobody}
The following statements hold:
\begin{enumerate}[a)]
	\item $H_0$ with domain $\D (H_0) = \{ f \in \H_2 |\ H_0 f \in \H_2 \}$ is self-adjoint. 
	\item \label{thm_selfadjointnessfreetwobody_b}$\Mminp$ with domain 
	$\D (\Mminp)=\{f\in\Hrel |\, \Mminp f \in \Hrel \}$ 
	is self-adjoint. 
	\item $\PMpl$ with domain 
	$\D (\PMpl) = \{f \in \Hcom |\, \PMpl f \in \Hcom \}$ 
	is self-adjoint. 
	\item $H_0$ with domain $\D_0$ is essentially self-adjoint. 
	\item \label{thm_selfadjointnessfreetwobody_e}$\Mminp$ with domain $\C^{16} \otimes H^2 (\R^3, \dd p)$ is essentially self-adjoint. 
\end{enumerate}
\end{theorem}

\begin{proof} 
Since self-adjointness of $H_0$, $\Mminp$, and $\PMpl$ is proven along the same lines, we only show it for $H_0$. Parts b) and c) then follow. 
\begin{enumerate}
	\item[a)] As the mass term $\beta m_1 \otimes \mathbf{1}_4 + \mathbf{1}_4 \otimes \beta m_2$ is bounded and symmetric, it suffices to show self-adjointness of $H_0$ for $m_1 = 0 = m_2$. In this case and recalling \eqref{eq_transformationtorelcomcoord} and \eqref{eq_def_T}, we obtain $T = U H_0 U^{-1}$. As $U$ is unitary, self-adjointness of $T$ implies self-adjointness of $H_0$. 
	
	We define the $\C^{16}\times\C^{16}$-matrix $T(\vec{P},\vec{p})$ for all $\vec{P}, \vec{p} \in \R^3$ by  
	\begin{align}
	T(\vec{P},\vec{p}) := \vec{P} \cdot \mathbf{M}^+ + \mathbf{M}^- \cdot \vec{p} .
	\end{align}
	Since $T(\vec{P},\vec{p})$ is Hermitian for all $\vec{P}, \vec{p} \in \R^3$, $T$ is symmetric on $\D (T) = U\D(H_0)$. Hence, it suffices to show $\D (T^*) \subseteq \D(T)$. 
	
	Let $f \in \D (T^*)$. Then, we obtain for all $g \in \D (T)$ 
	\begin{eqnarray}
	\langle T^* f, g \rangle = \langle f, T g \rangle 
	& = & 
	\int_{\R^3\times\R^3} f (\vec{P},\vec{p})^\dag \, T (\vec{P},\vec{p}) g (\vec{P},\vec{p}) \dd p \dd P \nonumber \\
	& = &
	\int_{\R^3\times\R^3} \left[ T (\vec{P},\vec{p}) f (\vec{P},\vec{p})\right]^\dag g (\vec{P},\vec{p}) \dd p \dd P
	\end{eqnarray}
	
	Since $(C^{\infty}_c (\R^3\times\R^3))^{\otimes 16}$ is contained in $\D (T)$, this holds in particular for all 
	$g \in (C^{\infty}_c (\R^3\times\R^3))^{\otimes 16}$. 
	This yields for almost all $\vec{P}, \vec{p} \in \R^3$ 
	\begin{align} 
	T (\vec{P},\vec{p}) f (\vec{P}, \vec{p}) 
	= 
	\left( T^* f \right) (\vec{P}, \vec{p}) ,
	\end{align}
	and thus, since 
	$f \in \D (T^{\ast})$, we get $T f \in U\H_2$. Therefore, we can conclude $f \in \D (T)$. 
	
	\item[d)] This is well-known, see e.g. {\cite[Corollary of Thm.~VIII.33, pp.~300]{RS1}}. 
	
	\item[e)] For any $f \in \Hrel$, we define the sequence $(f_n)_{n \in \mathbb{N}}$ as $f_n := \frac{n}{\hat{p}^2 + n} f$. Then, $f_n \in \C^{16} \otimes H^2 (\R^3, \dd r)$ for every fixed $n \in \mathbb{N}$ since 
	\begin{align} 
	\int_{\R^3}
	\left|
	(1 + p^2) \frac{n}{p^2 + n} \hat{f} (\vec{p})
	\right|^2
	\dd p
	\leq
	(1+n)
	\left\|	f \right\|^2 < \infty . 
	\end{align}
	Furthermore, $\norm{f_n - f}\xrightarrow{n \rightarrow \infty} 0$ by dominated convergence. For $f \in \D(\Mminp)$, we obtain similarly 
	\begin{align}
	\left\| \Mminp \left( f -  f_n \right) \right\|^2 
	=
	\int_{\R^3}
	\left|
	\left(1-\frac{n}{p^2 + n}\right) \mathbf{M}^- \cdot \vec{p} \hat{f} (\vec{p})
	\right|^2
	\dd p
	\xrightarrow{n \rightarrow \infty} 0 
	\end{align}
	by dominated convergence. Hence, $\C^{16} \otimes H^2 (\R^3, \dd r)$ is a core for $\Mminp$. \qedhere
\end{enumerate}
\end{proof}

The next relation in Lemma~\ref{lem_unsualstructure}\ref{lem_unsualstructure_a}) lays open the structure of $\D (\Mminp)$. Part~\ref{lem_unsualstructure_b}) of Lemma~\ref{lem_unsualstructure} below then shows that the restriction of $\Mminp$ to $\Hrel_+$ behaves as one would expect: It is self-adjoint on $\left( \C^{16} \otimes H^1 (\R^3, \dd r) \right) \cap \Hrel_+$. 

\begin{lemma}\label{lem_unsualstructure}
${}$
\begin{enumerate}[a)]
	\item \label{lem_unsualstructure_a} For all $f \in \mathcal{D} (\Mminp)$, we have 
	$\| - i \nabla_{\vec{r}} P_+ f \| = \left\| \frac{1}{2} \Mminp \, f \right\|$. 
	
	\item \label{lem_unsualstructure_b}
	The restriction of $\Mminp$ to $\Hrel_+$, denoted by $\Mminp \upharpoonright \mathcal{D}_+$, with domain
	$P_+ \mathcal{D} (\Mminp)$
	is self-adjoint and, recalling \eqref{eq_def_Dplus}, it holds 
	\begin{align}\label{eq_domainP_+partofMminp}
	\mathcal{D}_+ = 
	\left( \C^{16} \otimes H^1 (\R^3, \dd r) \right) \cap \Hrel_+ =
	P_+ \mathcal{D} (\Mminp) . 
	\end{align}
\end{enumerate}
\end{lemma}

\begin{proof}
\begin{enumerate}[a)]
	\item A computation shows that for almost all $\vec{p} \in \R^3$ we have $(\frac{1}{2} \mathbf{M}^- \cdot \vec{p})^2 = p^2 P_+ (\vec{p})$, and hence, $(\frac12 \Mminp)^2 = \hat{p}^2 P_+$ holds on the intersection of their domains. Thus, for all $g \in \C^{16} \otimes H^2 (\R^3, \dd r)$ we have
	\begin{eqnarray}\label{eq_equalitynormnablaPplusMminp}
	\left\| 
	- i \nabla_{\vec{r}} P_+ g 
	\right\|^2 
	& = & 
	\vphantom{\frac{1}{2}} 
	\left\langle 
	- i \nabla_{\vec{r}} P_+ g, - i \nabla_{\vec{r}} P_+ g
	\right\rangle 
	\overset{(\ast)}{=} 
	\left\langle 
	P_+ g, \hat{p}^2 P_+ g 
	\right\rangle \nonumber \\
	& = & 
	\left\langle 
	P_+ g, \left( \frac{1}{2} \Mminp \right)^2 g 
	\right\rangle 
	\overset{(\ast \ast)}{=} 
	\left\langle 
	g, \left( \frac{1}{2} \Mminp \right)^2 g 
	\right\rangle \nonumber\\
	& = & 
	\left\| 
	\frac{1}{2} \Mminp \, g 
	\right\|^2 ,
	\end{eqnarray}
	where we used in $(\ast)$ that the boundary terms vanish since 
	$P_+ g \in \C^{16} \otimes H^2 (\R^3, \dd r)$ 
	and in $(\ast \ast)$ that 
	$P_+ \Mminp \, g = \Mminp \, g$. 
	
	Hence, claim a) already holds on $\C^{16} \otimes H^2 (\R^3, \dd r)$, which however is a core for $\Mminp$ by Theorem~\ref{thm_selfadjointnessfreetwobody}\ref{thm_selfadjointnessfreetwobody_e}). This means that for all $f \in \D (\Mminp)$ there exists a sequence $(f_n)_{n\in\N} \subset \C^{16} \otimes H^2 (\R^3, \dd r)$ that converges to $f$ in the graph norm of $\Mminp$. By Eq.~\eqref{eq_equalitynormnablaPplusMminp}, $(f_n)_{n\in\N}$ is then also a Cauchy sequence in the graph norm of $-i\nabla_{\vec{r}} P_+$. Now, 
	$- i \nabla_{\vec{r}} P_+ \colon \mathcal{D} (-i \nabla_{\vec{r}} P_+) \rightarrow \mathcal{H}^{\mathrm{rel}}$ 
	with domain 
	$\D(-i\nabla_{\vec{r}} P_+) = \{ f \in \Hrel |\, P_+ f \in \C^{16} \otimes H^1 (\R^3, \dd r) \}$ is closed, as $P_+$ is bounded and both $P_+$ and $- i \nabla_{\vec{r}}$ are closed. Thus, $(f_n)_{n\in\N}$ converges then also to $f$ with respect to the graph norm of $- i \nabla_{\vec{r}} P_+$. In conclusion, for all $f\in\D(\Mminp)$ we have 
	\begin{align}
	\left\|
	\frac12 \Mminp f
	\right\|
	=
	\lim_{n\rightarrow\infty}
	\left\|
	\frac12 \Mminp f_n
	\right\|
	=
	\lim_{n\rightarrow\infty}
	\left\|
	-i\nabla_{\vec{r}} P_+ f_n
	\right\|
	=
	\left\|
	-i\nabla_{\vec{r}} P_+ f
	\right\| .
	\end{align}
	
	\item As the projections $P_\pm$ are tailor-made for $\Mminp$, we obtain in the same manner as for $H^{\mathrm{rel}}$ in line~\eqref{eq_defHrel} the matrix representation 
	\begin{align} 
	\Mminp= 
	\left(
	\begin{array}{cc}
	\Mminp \upharpoonright \mathcal{D}_+ \quad & \mathbf{0} \\
	\mathbf{0}  & \mathbf{0} 
	\end{array}
	\right) .
	\end{align}
	Now, by \cite[Prop.~2.6.3, p.~144]{tretter}, $\Mminp$ is self-adjoint if and only if $\Mminp \upharpoonright \mathcal{D}_+$ and $\mathbf{0}$ (that is, the operators in the upper left and lower right corner) are self-adjoint. Hence, self-adjointness of $\Mminp$, guaranteed by Theorem~\ref{thm_selfadjointnessfreetwobody}\ref{thm_selfadjointnessfreetwobody_b}), proves self-adjointness of $\Mminp \upharpoonright \D_+$. 
	
	It remains to prove Eq.~\eqref{eq_domainP_+partofMminp}. Suppose $f \in P_+ \D (\Mminp)$. Then, $f = P_+ g$ for some $g \in \D (\Mminp)$. Since $g \in \D (\Mminp)$ and $\Mminp$ and $P_+$ commute,
	\begin{align} 
	\| \Mminp f \| = \| \Mminp P_+ g \| \leq \| \Mminp g \| < \infty
	\end{align}
	and so $f \in \D (\Mminp) \cap \Hrel_+$. By part a),
	\begin{align} 
	\| - i \nabla_{\vec{r}} f \| = \| - i \nabla_{\vec{r}} P_+ f \| = \| \frac12 \Mminp f \| < \infty 
	\end{align}
	and so $f \in \C^{16} \otimes H^1 (\R^3, \dd r)$. Therefore, we have 
	$f \in \left( \C^{16} \otimes H^1 (\R^3, \dd r) \right) \cap \Hrel_+$. Suppose conversely that 
	$f \in \left( \C^{16} \otimes H^1 (\R^3, \dd r) \right) \cap \Hrel_+$. 
	Then, $f = P_+ f$ and 
	$f \in \C^{16} \otimes H^1 (\R^3, \dd r) \subseteq \mathcal{D} (\Mminp)$, 
	which implies $f \in P_+ \mathcal{D} (\Mminp)$. 
	\qedhere
\end{enumerate}
\end{proof}

\begin{remark}\label{rem_unsualstructure}
We want to remark the following. 
\begin{enumerate}[a)]
	\item The matrix representation in the preceding proof gives an explicit form of $\D (\Mminp)$, namely 
	\begin{align}
	\D (\Mminp) = \D_+ \oplus \Hrel_- .
	\end{align}
	
	\item The remarkable structure of $\D (\Mminp)$ is the following. Since $\D (\Mminp) = \D_+ \oplus \Hrel_-$ and $\D_+ = (\C^{16} \otimes H^1 (\R^3,\dd r)) \cap \Hrel_+$ holds, an $f \in \D (\Mminp)$ must have $H^1$-regularity only in $\Hrel_+$. In $\Hrel_-$ however, no regularity---besides being square-integrable, of course---is required. 
	
	This amounts to the following phenomenon. Let $\psi (1,2)$ denote a two-particle state which lies in $H^1 (\R^3, \dd R) \otimes \D (\Mminp)$. Then, it is possible to construct $\psi (1,2)$ in such a way that the kinetic energy of each single particle is infinite, i.e., formally $E_{\mathrm{kin}} [i] = \infty$ for $i = 1,2$. These infinities, however, cancel and their sum, i.e., the total kinetic energy, is finite: $E_{\mathrm{kin}} [1] + E_{\mathrm{kin}} [2] < \infty$. We refer the interested reader to a forthcoming publication. 
\end{enumerate}
\end{remark}

\subsection{Self-adjoint extension of $H_{\mathrm{DC}}$}\label{section_self-adjointextension}
At first, we show that many standard techniques from the perturbation theory of self-adjoint operators are not applicable, already in the case of $H^{\mathrm{rel}} = \Mminp + \gamma V$. 

\begin{lemma}\label{lemma_Vnotrelbounded}
Let $\gamma \neq 0$ and let $0 < \kappa \leq 1$ be fixed. Then, $\gamma V$ is not relatively bounded by $\Mminp$.
\end{lemma}

\begin{proof}
We chose a $\delta \in \R^+$ such that $0 < \delta < \kappa$. Next, we define $f \in \D (\Mminp)$ by 
\begin{align} 
f (\vec{r}) = 
(0, -1, 0, -1, 1, 0, 1, 0, 0, 1, 0, 1, -1, 0, -1, 0)^\top 
\cdot 
\frac{e^{- \left| \vec{r} \right|^2}}{|\vec{r}|^{3/2 - \delta}} .
\end{align}
Then, it holds that $\Mminp f = 0$ and so $f \in \Ker (\Mminp)$. But $f \not\in \D (\gamma V)$ as 
\begin{align}
\left\| \gamma V f \right\|^2
=
\int_{\R^3} 
\frac{\gamma^2}{|\vec{r}|^{3-2\delta+2\kappa}} e^{- 2 \left| \vec{r} \right|^2}
\dd r
\end{align}
does not converge except for $\gamma = 0$. This proves the lemma. 
\end{proof}

Next, we want to define a self-adjoint extension of $\HDC$. Before we can give its domain and action, we need to line out the various steps (or layers of self-adjointness) involved in the construction, and hence, the corresponding proof of this self-adjoint extension. 

As first step, we note that the $L^2$-parts of $\PMpl$ and of $H^{\mathrm{rel}} = \Mminp + \gamma V$ do not overlap, and therefore, we look only for a self-adjoint extension of $H^{\mathrm{rel}}$. This extension plus $\PMpl$ is then also a self-adjoint extension of $\HDC$. Since $H^{\mathrm{rel}}$ is a matrix operator, we have as second step the Frobenius-Schur factorization at our disposal. Therefore, we can use that important properties such as closedness and self-adjointness are encoded in the Schur complement. After a slight but important modification of $H^{\mathrm{rel}}$, we see that the Schur complement possesses a self-adjoint extension, denoted by $S_F$. This, in turn, paves the way to $H^{\mathrm{rel}}_F$, the self-adjoint extension of $H^{\mathrm{rel}}$, and finally, to $H_F$ which denotes the self-adjoint extension of $\HDC$ in relative and center-of-mass coordinates. 

We start with the just mentioned slight but important modification of $H^{\mathrm{rel}}$. We define the symmetric and bounded matrix 
\begin{align}\label{eq_defB}
B := 2 \beta \otimes \beta ,
\end{align}
recall Eq.~\eqref{eq_diracmatrices}, and consider $H^{\mathrm{rel}} + P_+ B P_+$. Recalling the matrix representation of $H^{\mathrm{rel}}$ from line~\eqref{eq_defHrel}, we find 
\begin{align}\label{eq_matrixrepresentationmodHrel}
H^{\mathrm{rel}} + P_+ B P_+
=
\left(
\begin{array}{cc}
\Mminp + P_+ B P_+ + P_+ \gamma V P_+ & \quad P_+ \gamma V P_-\\
P_- \gamma V P_+ & \quad P_- \gamma V P_-
\end{array}
\right)
\end{align}
which is well-defined on 
\begin{align}
\D (H^{\mathrm{rel}}) = \D_+ \oplus (\mathcal{D} (V) \cap \Hrel_-) ,
\end{align}
where $\D_+ = ( \C^{16} \otimes H^1 (\R^3, \dd r) ) \cap \Hrel_+$, introduced in \eqref{eq_def_domain_Hrel}. 

In the following lemma, we examine the properties of the combination of the matrices $P_+ (\vec{p}) B P_+ (\vec{p})$ and $\mathbf{M}^- \cdot \vec{p}$. 

\begin{lemma}\label{lemma_propertiesMminB}
The following statements hold for almost all $\vec{p} \in \R^3$: 
\begin{enumerate}[a)]
	\item \label{lemma_propertiesMminB_a}
	$B P_+ (\vec{p}) = P_+ (\vec{p}) B$, 
	
	\item \label{lemma_propertiesMminB_b}
	$ \left( \mathbf{M}^- \cdot \vec{p} + B P_+ (\vec{p}) \right)^2 = 4 \left( p^2 + 1 \right) P_+ (\vec{p})$, and 
	
	\item \label{lemma_propertiesMminB_c}
	$\left| \mathbf{M}^-  \cdot \vec{p} + B P_+ (\vec{p}) \right| = 2 \left( p^2 + 1 \right)^{1/2} P_+ (\vec{p})$. 
\end{enumerate}
\end{lemma}

\begin{proof}
\begin{enumerate}[a)]
	\item We get by direct computation $B P_+ (\vec{p}) = P_+ (\vec{p}) B$ for almost all $\vec{p} \in \R^3$ since $(\vec{\alpha} \cdot \vec{p}) \beta = - \beta (\vec{\alpha} \cdot \vec{p})$ by Eq.~\eqref{clifford_1}. 
	
	\item Using $P_+ (\vec{p}) \mathbf{M}^- \cdot \vec{p} = \mathbf{M}^- \cdot \vec{p} P_+ (\vec{p}) = \mathbf{M}^- \cdot \vec{p}$, we compute for almost all $\vec{p} \in \R^3$ $(\mathbf{M}^-\cdot \vec{p} + B P_+ (\vec{p}))^2 = 4 ( p^2 + 1 ) P_+ (\vec{p}) $ where again we made use of the anticommutation relations in Eq.~\eqref{clifford_1}. 
	
	\item By part \ref{lemma_propertiesMminB_a}), we have $P_+ (\vec{p}) B P_+ (\vec{p}) = B P_+ (\vec{p})$, and thus, $\mathbf{M}^- \cdot \vec{p} + B P_+ (\vec{p})$ is Hermitian. We compute for almost all $\vec{p} \in \R^3$ 
	\begin{align}
	\left| \mathbf{M}^- \cdot \vec{p} + B P_+ (\vec{p}) \right|
	& = 
	\left(
	\left( \mathbf{M}^- \cdot \vec{p} + B P_+ (\vec{p}) \right)^\dag \left( \mathbf{M}^- \cdot \vec{p} + B P_+ (\vec{p}) \right)
	\right)^{1/2} \nonumber \\
	& = 
	\left(
	4 \left( p^2 + 1 \right) P_+ (\vec{p}) 
	\right)^{1/2} 
	\end{align}
	where we used part \ref{lemma_propertiesMminB_b}) in the last step. Now, $P_+ (\vec{p})$ is a Hermitian matrix, and therefore, there exists a unitary matrix $u (\vec{p})$---the same matrix $u(\vec{p})$ as in lines~\eqref{eq_diagmatrixmrel} and \eqref{eq_diagmatrixPmin}---such that we have for almost all $\vec{p} \in \R^3$
	\begin{align}
	u (\vec{p}) P_+ (\vec{p}) u (\vec{p})^\dag 
	= 
	\left(
	\begin{array}{cc}
	\mathbf{1}_8 & \\
	& \mathbf{0}_8
	\end{array}
	\right) .
	\end{align}
	Then, 
	\begin{align}
	\left(
	4 \left( p^2 + 1 \right) P_+ (\vec{p}) 
	\right)^{1/2} 
	& = 
	u (\vec{p})^\dag 
	\;
	2 \left( p^2 + 1 \right)^{1/2}
	\left(
	\begin{array}{cc}
	\mathbf{1}_8 & \\
	& \mathbf{0}_8
	\end{array}
	\right)
	u (\vec{p}) \nonumber \\
	& = 
	2 \left( p^2 + 1 \right)^{1/2} P_+ (\vec{p}) 
	\end{align}
	for almost all $\vec{p} \in \R^3$ which concludes the proof. 
	\qedhere
\end{enumerate}
\end{proof}

In the following, we will make frequent use of the operator $\Mminp + B P_+$ in the underlying Hilbert space $\Hrel_+$. Thus, we introduce the abbreviation 
\begin{align}
A_0 \equiv \Mminp + B P_+ .
\end{align}
First, we want to relate the different interaction potentials, distinguished by the exponent $0 < \kappa \leq 1$, to $A_0$. Recall $\D_+ = ( \C^{16} \otimes H^1 (\R^3, \dd r) ) \cap \Hrel_+$ from line~\eqref{eq_def_Dplus}. 

\begin{lemma}\label{lemma_modifiedHerbst}
Let $0 < \kappa \leq 1$ and let $M_{\kappa} > 0$ be given by
\begin{align}
M_{\kappa} 
= 
2^{- \kappa}  \frac{\Gamma \left( \frac{3}{4} - \frac{\kappa}{2} \right)}{\Gamma \left( \frac{3}{4} + \frac{\kappa}{2} \right)} . 
\end{align}
\begin{enumerate}[a)]
	\item \label{lemma_modifiedHerbst_a}
	For all $f \in \D_+$, it holds that
	\begin{align}
	\left\| | \cdot |^{-\kappa} f \right\|
	\leq 
	\frac{M_{\kappa}}{2} \left\| A_0 f \right\| .
	\end{align}
	
	\item \label{lemma_modifiedHerbst_b}
	For all $f \in ( \C^{16} \otimes H^{1/2} (\R^3, \dd r) ) \cap \Hrel_+$, it holds that
	\begin{align}
	\left\langle f, | \cdot |^{-\kappa} f \right\rangle
	\leq 
	\frac{M_{\kappa/2}^2}{2} \left\langle f, \left| A_0 \right| f \right\rangle .
	\end{align}
\end{enumerate}
\end{lemma}

\begin{proof}
\begin{enumerate}[a)]
	\item By {\cite[Theorem~2.5]{Herbst1977}}, for all $h \in \C^{16} \otimes \mathcal{S} (\R^3)$ we have the bound 
	\begin{align}\label{eq_herbstineq}
	\left\| 
	| \cdot |^{-\kappa} (\hat{p}^2 + 1)^{- \kappa / 2} h 
	\right\| 
	\leq 
	M_{\kappa} \| h \| . 
	\end{align}
	This inequality extends to all $h \in \Hrel_+$. Now, for all $f \in \D_+$, there exists a $g \in \Hrel_+$ such that $f = (\hat{p}^2 + 1)^{-\kappa/2} g$. Thus, for all $f \in \D_+$ we can compute 
	\begin{align}\label{eq_modifiedHerbst_1}
	\left\|
	|\cdot|^{-\kappa} f 
	\right\|
	=
	\left\|
	|\cdot|^{-\kappa} (\hat{p}^2 + 1)^{-\kappa/2} g 
	\right\| 
	\leq
	M_\kappa \| g \| 
	=
	M_\kappa 
	\left\|
	(\hat{p}^2 + 1)^{\kappa/2} f
	\right\| .
	\end{align}
	With part \ref{lemma_propertiesMminB_c}) of Lemma~\ref{lemma_propertiesMminB}, for all $f \in \D_+$ we obtain 
	\begin{align}\label{eq_modifiedHerbst_2}
	\left\|
	|\cdot|^{-\kappa} f 
	\right\|
	& \leq 
	\frac{M_\kappa}{2}
	\left\|
	2 (\hat{p}^2 + 1)^{\kappa/2} f
	\right\| 
	\leq
	\frac{M_\kappa}{2}
	\left\|
	2 (\hat{p}^2 + 1)^{1/2} P_+ f
	\right\| \nonumber \\
	& = 
	\frac{M_\kappa}{2}
	\left\|
	\left| \Mminp + B P_+ \right| f
	\right\| 
	=
	\frac{M_\kappa}{2}
	\left\|
	A_0 f
	\right\| ,
	\end{align}
	which proves the claim. 
	
	\item Due to \eqref{eq_modifiedHerbst_1}, for all $f \in \D_+$ we find 
	\begin{align}
	\left\langle
	f, |\cdot|^{-\kappa} f
	\right\rangle
	& = 
	\left\|
	|\cdot|^{-\kappa/2} f
	\right\|^2 
	\leq
	M_{\kappa/2}^2
	\left\|
	(\hat{p}^2 + 1)^{\kappa/4} P_+ f
	\right\|^2 \nonumber \\
	& = 
	M_{\kappa/2}^2
	\left\langle
	f, (\hat{p}^2 + 1)^{1/2} P_+ f
	\right\rangle 
	=
	\frac{M_{\kappa/2}^2}{2}
	\left\langle
	f, 2 (\hat{p}^2 + 1)^{1/2} P_+ f
	\right\rangle \nonumber \\
	& = 
	\frac{M_{\kappa/2}^2}{2}
	\left\langle
	f, \left| A_0 \right| f 
	\right\rangle .
	\end{align}
	This computation extends to all $f \in ( \C^{16} \otimes H^{1/2} (\R^3, \dd r) ) \cap \Hrel_+$ since $\C^{16} \otimes H^{1} (\R^3, \dd r)$ is a core of $\langle	\cdot, (\hat{p}^2 + 1)^{1/2} \cdot \rangle$ (see \cite[Theorem~7.14]{liebloss}). 
	\qedhere
\end{enumerate}
\end{proof}

Part \ref{lemma_propertiesMminB_a}) of Lemma~\ref{lemma_propertiesMminB} says that $B$ leaves $\Hrel_+$ invariant. Consequently, we can define the operator
\begin{align}
A :=
A_0 + P_+ \gamma V P_+ 
\end{align}
in the underlying Hilbert space $\Hrel_+$ with domain 
\begin{align}
\D (A) = \D_+ = \left( \C^{16} \otimes H^1 (\R^3, \dd r) \right) \cap \Hrel_+ .
\end{align}

The next lemma provides important properties of $A$ in order to obtain a self-adjoint extension of $H^{\mathrm{rel}} + P_+ B P_+$ later on in Theorem~\ref{thm_self-adjointextensionHrel}. 

\begin{lemma}\label{lemma_Asa0resolvent}
Let $0 < \kappa \leq 1$ and let $M_\kappa$ be given as in Lemma~\ref{lemma_modifiedHerbst}. Moreover, let $|\gamma| M_\kappa < 2$. Then, 
\begin{enumerate}[a)]
	\item $A$ is self-adjoint on $\D_+$, and 
	
	\item $0 \in \rho (A)$. 
\end{enumerate}
\end{lemma}

\begin{proof}
\begin{enumerate}[a)]
	\item Since $B P_+$ is symmetric and bounded, $A_0$ is self-adjoint on $\D_{+}$ by Lemma~\ref{lem_unsualstructure}\ref{lem_unsualstructure_b}). We have for all $f \in \D_+$ 
	\begin{align}\label{eq_relativeboundedness}
	\left\| P_+ \gamma V P_+ f \right\| 
	\leq 
	\left\| \gamma V f \right\| 
	=
	\vphantom{\frac{1}{2}} 
	\left\| 
	\frac{\gamma}{| \cdot |^\kappa} f 
	\right\| 
	\leq 
	\vphantom{\frac{1}{2}} 
	\frac{|\gamma| M_\kappa}{2}
	\left\|
	A_0 f
	\right\| 
	\end{align}
	where we used Lemma~\ref{lemma_modifiedHerbst}\ref{lemma_modifiedHerbst_a}) in the last estimate. For $|\gamma| M_\kappa  < 2$, the Kato-Rellich theorem now implies self-adjointness of $A$ on $\D_+$. 
	
	\item First, we prove that $A_0$ has a bounded inverse $A_0^{-1} \colon \Hrel_+ \rightarrow \D_+$. For all $f \in \Hrel_+$, we find $(\hat{p}^2 + 1)^{-1} f \in \C^{16} \otimes H^2 (\R^3, \dd r)$ and $A_0 (\hat{p}^2 + 1)^{-1} f \in \D_+$ with 
	\begin{align}\label{eq_norminverse}
	\left\|
	A_0 (\hat{p}^2 + 1)^{-1}
	\right\|
	=
	\sup_{\vec{p}\in \R^3} \left| \frac{\mathbf{M}^- \cdot \vec{p}+ B P_+ (\vec{p})}{p^2 + 1} \right| 
	= 
	2
	\end{align}
	By Lemma~\ref{lemma_propertiesMminB}\ref{lemma_propertiesMminB_b}), on $(\C^{16} \otimes H^2 (\R^3, \dd r))\cap\Hrel_+$ we have $A_0^2 = 4 (\hat{p}^2 + 1) P_+$. Hence, for all $f \in \Hrel_+$ 
	\begin{align}
	A_0 
	\left(
	\frac14 A_0 (\hat{p}^2 + 1)^{-1} f
	\right)
	= f
	\end{align}
	which together with Eq.~\eqref{eq_norminverse} implies that $A_0^{-1} = A_0 (\hat{p}^2 + 1)^{-1}/4$. 
	
	Moreover, for $|\gamma| M_\kappa < 2$ a theorem by Kato {\cite[Theorem~IV.1.16, p.~196]{kato_perturbation}} in combination with Eq.~\eqref{eq_relativeboundedness} gives the existence of a bounded inverse of $A$. The same theorem by Kato also implies that $A \D (A)=\Hrel_+$ (see {\cite[Lemma~1]{schmincke_ess-sa}}), which in return implies $0 \in \rho (A)$. 
	\qedhere
\end{enumerate}
\end{proof}

Next, we aim at the Frobenius-Schur factorization of $H^{\mathrm{rel}} + P_+ B P_+$. Recalling the matrix representation of $H^{\mathrm{rel}} + P_+ B P_+$ from line~\eqref{eq_matrixrepresentationmodHrel}, we find 
\begin{align}\label{eq_HrelplusBinmatrixform}
H^{\mathrm{rel}} + P_+ B P_+ = 
\left(
\begin{array}{cc}
A & \quad P_+ \gamma V P_- \\
P_- \gamma V P_+ & \quad P_- \gamma V P_-
\end{array}
\right) 
\end{align}
with domains 
\begin{align}\label{eq_defdomainsHrelplusBinmatrixform}
\begin{aligned}
A \colon \D_+ &\rightarrow \Hrel_+ , \quad 
P_+ \gamma V P_- \colon \D (V) \cap \Hrel_- \rightarrow \Hrel_+ \\
P_- \gamma V P_+ \colon \D (V) \cap \Hrel_+ &\rightarrow \Hrel_- , \quad 
P_- \gamma V P_- \colon \D (V) \cap \Hrel_- \rightarrow \Hrel_- 
\end{aligned}
\end{align}
and define the Schur complement 
$S \colon \D (S) \rightarrow \Hrel_-$
of $A$ by 
\begin{align}
S := P_- \gamma V P_- - P_- \gamma V P_+ A^{- 1} P_+ \gamma V P_- 
\end{align}
with domain 
\begin{align}
\D (S) = \D (V) \cap \Hrel_- . 
\end{align}
That $S$ is well-defined, and thus the Frobenius-Schur factorization of $H^{\mathrm{rel}} + P_+ B P_+$ exists, is the content of the next lemma. 

\begin{lemma}\label{lemma_frobschurofHrel}
Let $0 < \kappa \leq 1$ and let $|\gamma| M_\kappa < 2$. Then, the matrix representation of $H^{\mathrm{rel}} + P_+ B P_+$ is symmetric and its Frobenius-Schur factorization is given by
\begin{align} 
H^{\mathrm{rel}} + P_+ B P_+ 
= 
\left(
\begin{array}{cc}
\mathrm{id} & \mathbf{0} \\
P_- \gamma V P_+ A^{- 1} \quad & \mathrm{id}
\end{array}
\right) 
\left(
\begin{array}{cc}
A & \mathbf{0} \\
\mathbf{0} & S 
\end{array}
\right) 
\left(
\begin{array}{cc}
\mathrm{id} & \quad A^{- 1} P_+ \gamma V P_- \\
\mathbf{0} & \mathrm{id} 
\end{array}
\right) . 
\end{align}
\end{lemma}

\begin{proof}
We use the theory of unbounded matrix operators in Hilbert space, which we discuss in Appendix~\ref{section_matrixoperators}. There we work with the assumptions~(A\ref{assumptions_matrixop_1})-(A\ref{assumptions_matrixop_6}), which we need to check now: 
\begin{enumerate}[({A}1)]
	\item $A$, $P_+ \gamma V P_-$, $P_- \gamma V P_+$, and $P_- \gamma V P_-$ have dense domains by Proposition~\ref{prop_domainH^reldense}. 
	As they are both symmetric, $A$ and $P_- \gamma V P_-$ are closable. 
	$P_+ \gamma V P_-$ is closable as $(P_- \gamma V P_+)^*$ is a closed extension of it and vice versa which can be seen as follows: 
	Suppose that 
	$f \in \D (V) \cap \Hrel_-$ and 
	$g \in \D (V) \cap \Hrel_+$. 
	Then, 
	\begin{align}
	\left\langle 
	f , P_- V P_+ g
	\right\rangle 
	= 
	\left\langle 
	f , V g
	\right\rangle 
	= 
	\left\langle 
	V f , g
	\right\rangle 
	=
	\left\langle 
	P_+ V P_- f , g 
	\right\rangle ,
	\end{align}
	and thus, $f \in \D ((P_- \gamma V P_+)^*)$. 
	
	\item $\D (P_+ \gamma V P_-) =\D (P_- \gamma V P_-)$ by definition in \eqref{eq_defdomainsHrelplusBinmatrixform}. 
	
	\item The resolvent set of $A$ is not empty as $0 \in \rho (A)$ by Lemma~\ref{lemma_Asa0resolvent}. 
	
	\item $\D (A^{\ast}) = \D_+ \subset \D ((P_+ \gamma V P_-)^{\ast})$ by definition in \eqref{eq_defdomainsHrelplusBinmatrixform} and Hardy's inequality. 
	
	\item $\D (A) \subset \D (P_- \gamma V P_+)$ by line~\eqref{eq_defdomainsHrelplusBinmatrixform} and Hardy's inequality. 
	
	\item $\D (H^{\mathrm{rel}} + P_+ B P_+) =\D_+ \oplus ( \D (V) \cap \Hrel_-) $ which is dense in $\Hrel$ by Lemma~\ref{prop_domainH^reldense}. 
\end{enumerate}

Symmetry follows from Lemma~\ref{lemma_symmetryblockmatrix} (Appendix~\ref{section_matrixoperators}) with the help of Lemma~\ref{lemma_Asa0resolvent}. Existence of the Frobenius-Schur factorization follows from Theorem~\ref{thm_closability} (Appendix~\ref{section_matrixoperators}). 
\end{proof}

As outlined in the beginning of this section, the crucial ingredient in finding a self-adjoint extension of $\HDC$ is to find a self-adjoint extension of the Schur complement $S$. However, before we can state this extension in Lemma~\ref{lemma_schursym} below, we need to provide the technical results of Theorem~\ref{thm_coreV1/2} and Lemma~\ref{lemma_boundedoperatorforformboundedness}. The proof of Theorem~\ref{thm_coreV1/2} is postponed to the next Section~\ref{section_prooftechnicalstuff}. 

\begin{theorem}\label{thm_coreV1/2}
Let $0 < \kappa \leq 1$. For every 	$f \in \D (V^{1/2}) \cap \Hrel_-$, 	there exists a sequence $(f_n)_{n \in \N} \subset \D (V)$	such that
\begin{align}
\left\| 
f - P_- f_n 
\right\| 
+ 
\left\| 
V^{1/2} P_- (f - f_n) 
\right\| 
\xrightarrow{n \rightarrow \infty} 0 . 
\end{align}
\end{theorem}

\begin{lemma}\label{lemma_boundedoperatorforformboundedness}
Let $0 < \kappa \leq 1$ and let $|\gamma| M_{\kappa/2}^2 < 1$. Then, the operator 
\begin{align}
(|\gamma| V)^{1/2} P_+ A^{-1} P_+ (|\gamma| V)^{1/2} \colon \D (V^{1/2}) \rightarrow \Hrel
\end{align}
is bounded on the dense set 
$\D (V^{1/2})$ 
with 
\begin{align}
C := \| (|\gamma| V)^{1/2} P_+ A^{-1} P_+ (|\gamma| V)^{1/2} \| < 1 .
\end{align}
\end{lemma}

\begin{proof}
As preparation, we prove that the operator 
$
(|\gamma| V)^{1/2} P_+ |A_0|^{-1/2} \colon \Hrel_+ \rightarrow \Hrel
$
is bounded with norm less or equal to $\sqrt{|\gamma|/2} \, M_{\kappa/2}$. As $|A_0|^{-1/2}$ maps into $\Hrel_+$, this is equivalent to 
$
\left\| (|\gamma| V)^{1/2} |A_0|^{-1/2} \right\| \leq \sqrt{|\gamma|/2} \, M_{\kappa/2} ,
$.
By Lemma~\ref{lemma_modifiedHerbst}\ref{lemma_modifiedHerbst_b}), for all $f \in (\C^{16} \otimes H^{1/2} (\R^3, \dd r)) \cap \Hrel_+$ we get 
\begin{align}
\left\langle f, |\gamma| V f \right\rangle 
=
|\gamma|
\left\langle f, |\cdot|^{-\kappa} f \right\rangle 
\leq
\frac{|\gamma| M_{\kappa/2}^2}{2} \left\langle f , \left| A_0 \right| f \right\rangle 
\end{align}
or equivalently 
\begin{align}
\left\| (|\gamma| V)^{1/2} f \right\|
\leq 
\sqrt{\frac{|\gamma|}{2}} \, M_{\kappa/2}
\left\| \left| A_0 \right|^{1/2} f \right\| .
\end{align}
As $\left| A_0 \right|^{-1/2}$ maps $\Hrel_+$ into $(\C^{16} \otimes H^{1/2} (\R^3, \dd r) ) \cap \Hrel_+$, this implies for all $g \in \Hrel_+$ 
\begin{align}
\left\| (|\gamma| V)^{1/2} \left| A_0 \right|^{-1/2} g \right\|
\leq 
\sqrt{\frac{|\gamma|}{2}} \, M_{\kappa/2}
\left\| g \right\| .
\end{align}
This also implies 
\begin{align}
\left\| \left| A_0 \right|^{-1/2} P_+ (|\gamma| V)^{1/2} \right\|
& = 
\left\| \left( \left| A_0 \right|^{-1/2} P_+ (|\gamma| V)^{1/2} \right)^*\right\| \nonumber \\
& =  
\left\| \left( (|\gamma| V)^{1/2} \right)^* P_+^* \left( \left| A_0 \right|^{-1/2} \right)^* \right\| \nonumber \\
& =  
\left\| (|\gamma| V)^{1/2} P_+ \left| A_0 \right|^{-1/2} \right\| 
\leq
\sqrt{\frac{|\gamma|}{2}} \, M_{\kappa/2} .
\end{align}
We use the polar decomposition $A_0 = U_{A_0} |A_0|$. Note that $\Ker (A_0) = \{ 0 \}$, and hence, $U_{A_0}$ is unitary. We obtain 
\begin{align}
&\left\| (|\gamma| V)^{1/2} P_+ A_0^{-1} P_+ (|\gamma| V)^{1/2} \right\| \nonumber \\
&=
\left\| (|\gamma| V)^{1/2} P_+ \left( U_{A_0} |A_0| \right)^{-1} P_+ (|\gamma| V)^{1/2} \right\| \nonumber \\
& = 
\left\| (|\gamma| V)^{1/2} P_+ |A_0|^{-1} U_{A_0}^{-1} P_+ (|\gamma| V)^{1/2} \right\| \nonumber \\
& =  
\left\| (|\gamma| V)^{1/2} P_+ |A_0|^{-1/2} U_{A_0}^{-1} |A_0|^{-1/2} P_+ (|\gamma| V)^{1/2} \right\| \nonumber \\
& \leq 
\left\| (|\gamma| V)^{1/2} P_+ |A_0|^{-1/2} \right\|^2 \left\| U_{A_0}^{-1} \right\| 
\leq
\frac{|\gamma| M_{\kappa/2}^2}{2}
\end{align}
where we used unitarity of $U_{A_0}^{-1}$ and the fact that $U_{A_0}^{-1}$ and $|A_0|^{-1/2}$ commute as both are functions of the self-adjoint operator $A_0$.

With the help of the resolvent identity we find 
\begin{align}
& \left\| (|\gamma| V)^{1/2} P_+ A^{-1} P_+ (|\gamma| V)^{1/2} \right\| \nonumber \\
& \leq 
\left\| (|\gamma| V)^{1/2} P_+ \left( A^{-1} - A_0^{-1} \right) P_+ (|\gamma| V)^{1/2} \right\| 
+ 
\left\| (|\gamma| V)^{1/2} P_+ A_0^{-1} P_+ (|\gamma| V)^{1/2} \right\| \nonumber \\
& \leq  
\left\| (|\gamma| V)^{1/2} P_+ A^{-1} ( A_0 - A ) A_0^{-1}  P_+ (|\gamma| V)^{1/2} \right\| + \frac{|\gamma| M_{\kappa/2}^2}{2} \nonumber \\ 
& = 
\left\| (|\gamma| V)^{1/2} P_+ A^{-1} P_+ |\gamma| V P_+ A_0^{-1}  P_+ (|\gamma| V)^{1/2} \right\| + \frac{|\gamma| M_{\kappa/2}^2}{2} \nonumber \\
& = 
\left\| 
\left( 
(|\gamma| V)^{1/2} P_+ A^{-1} P_+ (|\gamma| V)^{1/2}
\right)
\left(
(|\gamma| V)^{1/2} P_+ A_0^{-1}  P_+ (|\gamma| V)^{1/2} 
\right)
\right\| + \frac{|\gamma| M_{\kappa/2}^2}{2} \nonumber \\ 
& \leq 
\frac{|\gamma| M_{\kappa/2}^2}{2}
\left\| (|\gamma| V)^{1/2} P_+ A^{-1} P_+ (|\gamma| V)^{1/2} \right\|
+ \frac{|\gamma| M_{\kappa/2}^2}{2}
\end{align}
which implies 
\begin{align}
C
=
\left\| (|\gamma| V)^{1/2} P_+ A^{-1} P_+ (|\gamma| V)^{1/2} \right\|
\leq 
\frac{|\gamma| M_{\kappa/2}^2}{2 - |\gamma| M_{\kappa/2}^2} . 
\end{align}
If $|\gamma| M_{\kappa/2}^2 < 1$, then $C < 1$. This concludes the proof. 
\end{proof}

\begin{lemma}\label{lemma_schursym}
Let $0 < \kappa \leq 1$ and let $|\gamma| M_{\kappa/2}^2 < 1$. Then, the form sum 
\begin{align}
S_F := \gamma (V^{1/2} P_-)^* V^{1/2} P_- + P_- \gamma V P_+ A^{- 1} P_+ \gamma V P_- 
\end{align}
with domain 
\begin{align}
\D (S_F) = \D (V^{1/2}) \cap \D (S^*) 
\end{align}
defines a self-adjoint extension of $S$. 
\end{lemma}

\begin{proof}
In what follows, it is very useful to distinguish the different scalar products of the underlying Hilbert spaces $\Hrel$ and $\Hrel_-$, respectively. We define two forms that map from $\Hrel_- \times \Hrel_-$ into $\C$. First, we define the form of $P_- V P_-$ by 
\begin{align}
\mathfrak{v} [f,g] := \left\langle f, P_- V P_- g \right\rangle_{\Hrel_-} \, , \quad 
\D (\mathfrak{v}) = \D (V) \cap \Hrel_- .
\end{align}
Second, we define the form $\mathfrak{t}$ by 
\begin{align}
\mathfrak{t} [f,g] := \left\langle V^{1/2} P_- f, V^{1/2} P_- g \right\rangle_{\Hrel} \, , \quad 
\D (\mathfrak{t}) = \D (V^{1/2}) \cap \Hrel_- .
\end{align}
Furthermore, we define the restriction of $V^{1/2}$ to $\Hrel_-$. In order to clearly distinguish it from $V^{1/2}$, we write it as $V^{1/2} P_-$, i.e., 
\begin{align}
V^{1/2} P_- \colon \D (V^{1/2}) \cap \Hrel_- \rightarrow \Hrel . 
\end{align}
As it is a map from $\Hrel_-$ to $\Hrel$, its adjoint $(V^{1/2} P_-)^*$ maps from $\Hrel$ back to $\Hrel_-$. We list a number of properties. 

\begin{enumerate}[(i)]
	\item $\mathfrak{v}$, $\mathfrak{t}$ are symmetric as $V$ and $V^{1/2}$ are self-adjoint in the underlying Hilbert space $\Hrel$. 
	
	\item $\mathfrak{v}$, $\mathfrak{t}$ are positive as $V$ and $V^{1/2}$ are positive. 
	
	\item $V^{1/2} P_-$ is closed as $V^{1/2}$ and $P_-$ are closed and $P_-$ is bounded. Thus, $\mathfrak{t}$ is closed. 
\end{enumerate}

We claim that $\overline{\mathfrak{v}} = \mathfrak{t}$. For all $f,g \in \D (V) \cap \Hrel_-$, we compute  
\begin{align}
\mathfrak{v} [f,g] 
& = 
\left\langle f, P_- V P_- g \right\rangle_{\Hrel_-}
=
\left\langle P_- f,  V P_- g \right\rangle_{\Hrel}
=
\left\langle P_- f,  V^{1/2} V^{1/2} P_- g \right\rangle_{\Hrel} \nonumber \\
& =  
\left\langle V^{1/2} P_- f,  V^{1/2} P_- g \right\rangle_{\Hrel} 
= 
\mathfrak{t} [f,g] 
\end{align}
where we used self-adjointness of $P_-$ and $V^{1/2}$. Hence, $\mathfrak{v}$ coincides with $\mathfrak{t}$ on $\D (V) \cap \Hrel_-$. Since $\D (V) \cap \Hrel_-$ is a form core for $\mathfrak{t}$ by Theorem~\ref{thm_coreV1/2}, we can conclude that $\overline{\mathfrak{v}} = \mathfrak{t}$. 

Now, by \cite[Theorem~VIII.15]{RS1}, $\mathfrak{t}$ is the form of a unique self-adjoint operator, denoted by $V_F$ with domain $\D (V_F)$. In particular, $V_F$ is self-adjoint in the underlying Hilbert space $\Hrel_-$. Since $P_- V P_-$ is positive and symmetric on $\D (V) \cap \Hrel_-$ as $V$ is positive and self-adjoint on $\D(V)$, also $\overline{\mathfrak{v}}$ is the form of a unique self-adjoint operator by \cite[Theorem~X.23]{RS2}. By $\overline{\mathfrak{v}} = \mathfrak{t}$ from above, we know that the operator associated with $\overline{\mathfrak{v}}$ is $V_F$. $V_F$ is the Friedrichs extension of $P_- V P_-$ and it is the unique self-adjoint extension whose domain is contained in $\D(\overline{\mathfrak{v}})$. Moreover, by \cite[Theorem~10.17]{schmuedgen}, we have 
\begin{align}
\D (V_F) 
& = 
\D (\overline{\mathfrak{v}}) \cap \D ((P_- V P_-)^*) 
=
\D (\mathfrak{t}) \cap \D ((P_- V P_-)^*) \nonumber \\
& = 
\D (V^{1/2}) \cap \D ((P_- V P_-)^*) . 
\end{align}

We compute $V_F$. For all $f \in \D (\mathfrak{t})$ and $g \in \D (V_F)$ we obtain 
\begin{align}
\left\langle V^{1/2} P_- f, V^{1/2} P_- g \right\rangle_{\Hrel}
=
\mathfrak{t} [f,g]
=
\left\langle f, V_F g \right\rangle_{\Hrel_-} .
\end{align}
Therefore, $V^{1/2} P_- g$ lies in $\D ((V^{1/2} P_-)^*)$. Furthermore, density of $\D (\mathfrak{t})$ implies that 
$(V^{1/2} P_-)^* V^{1/2} P_- g = V_F g$ holds for all $g \in \D (V_F)$. We obtain
$V_F \subseteq (V^{1/2} P_-)^* V^{1/2} P_-$. Self-adjointness of $V_F$ and symmetry of $(V^{1/2} P_-)^* V^{1/2} P_-$ imply $V_F = (V^{1/2} P_-)^* V^{1/2} P_-$. 

In order to connect the above to $S$, we estimate for all $f \in \D (V) \cap \Hrel_-$ 
\begin{align}\label{eq_formboundV_F}
& \left| \left\langle f, P_- V P_+ \gamma A^{- 1} P_+ V P_- f \right\rangle_{\Hrel_-} \right| \nonumber \\
& =  
\left| \left\langle V^{1/2} P_- f, 
\left( 
V^{1/2} P_+ \gamma A^{- 1} P_+ V^{1/2}
\right)
V^{1/2} P_- f \right\rangle_{\Hrel} \right| \nonumber \\
& \leq 
\left\| V^{1 / 2} P_- f \right\|  \left\| V^{1 / 2} P_+ \gamma A^{- 1} P_+ V^{1 / 2} \right\|  \left\| V^{1 / 2} P_- f \right\| \nonumber \\
& =  
C \left\langle f, P_- V P_- f \right\rangle_{\Hrel_-} 
=
C \left\langle f, V_F f \right\rangle_{\Hrel_-} 
\end{align}
where $C < 1$ holds by Lemma~\ref{lemma_boundedoperatorforformboundedness}. Again by Theorem~\ref{thm_coreV1/2}, we know that $\D (V) \cap \Hrel_-$ is a form core for $\mathfrak{t}$, and thus, the inequality \eqref{eq_formboundV_F} also holds for all $f \in \D (V_F)$. 

The KLMN-theorem guarantees that the form sum $V_F + P_- V P_+ \gamma A^{- 1} P_+ V P_-$ is a self-adjoint operator with domain $\D (V_F)$. Moreover, as $\gamma$ is real, we know that the form sum $S_F := \gamma V_F + P_- \gamma V P_+ A^{- 1} P_+ \gamma V P_-$ with domain $\D(S_F)=\D(V_F)$ is a self-adjoint extension of $S$. 

It remains to show that $\D (S_F) = \D (V^{1/2}) \cap \D (S^*)$. To that end, we introduce the abbreviation $K \equiv V - V P_+ \gamma A^{- 1} P_+ V$ such that $S = P_- \gamma K P_-$. Next, we define the form $\mathfrak{k}$ by 
\begin{align}
\mathfrak{k} [f,g] := \left\langle f, P_- K P_- g \right\rangle_{\Hrel_-} , \, \quad \D (\mathfrak{k}) = \D (V) \cap \Hrel_- .
\end{align}

From the inequality in line~\eqref{eq_formboundV_F} and positivity of $P_- V P_-$, we conclude that $\mathfrak{k} [f,f] \geq 0$ for all $f \in \D (\mathfrak{k})$. Thus, by \cite[Proposition~10.4]{schmuedgen}, $K \geq 0$ also holds. Self-adjointness of $V$ on $\D (V)$ and the fact that $P_+ A^{-1} P_+$ maps into $\D (V)$ imply symmetry of $P_- K P_-$ on $\D (V) \cap \Hrel_-$. Therefore, we know that $P_- K P_-$ has a self-adjoint extension, its Friedrichs extension, denoted by $K_F$ with domain $\D (K_F) = \D (\overline{\mathfrak{k}}) \cap \D ((P_- K P_-)^*)$. Since $P_- K P_-$ and $S$ differ only by the real multiple $\gamma$, we get $\D (K_F) = \D (\overline{\mathfrak{k}}) \cap \D (S^*)$. 

Now, both $V_F + P_- V P_+ \gamma A^{- 1} P_+ V P_-$ and $K_F$ extend $P_- K P_-$. As they are uniquely distinguished by their respective form domains, we can conclude that they are equal if $\D (\overline{\mathfrak{v}}) = \D (\overline{\mathfrak{k}})$. In order to prove precisely that, we compute for all $f \in \D (V) \cap \Hrel_-$  
\begin{align}\label{eq_equivalenceforms}
\mathfrak{v} [f,f] 
\leq{}& 
\left\langle 
f, P_- V P_- f 
\right\rangle \nonumber \\
&{}+ 
\frac{1}{1 - C} 
\left( 
C \left\langle f, P_- V P_- f \right\rangle - \left\langle f, P_- V P_+ \gamma A^{- 1} P_+ V P_- f \right\rangle
\right) \nonumber \\
= {}& 
\frac{1}{1 - C}
\left(
\left\langle f, P_- V P_- f \right\rangle - \left\langle f, P_- V P_+ \gamma A^{- 1} P_+ V P_- f \right\rangle
\right) \nonumber \\
= {}& 
\frac{1}{1 - C} \, \mathfrak{k} [f,f] 
\leq 
\frac{2}{1 - C} \, \left\langle f, P_- V P_- f \right\rangle 
= 
\frac{2}{1 - C} \, \mathfrak{v} [f,f]
\end{align}
where we used \eqref{eq_formboundV_F} in the first and second to last step. This shows that $\D (\overline{\mathfrak{v}}) = \D (\overline{\mathfrak{k}})$ and thus $K_F = V_F + P_- V P_+ \gamma A^{- 1} P_+ V P_-$. We can now conclude 
\begin{align}
\D (S_F) = \D (K_F) = \D (\overline{\mathfrak{v}}) \cap \D (S^*) = \D (V^{1/2}) \cap \D (S^*) ,
\end{align}
which finishes the proof. 
\end{proof}

\begin{theorem}\label{thm_self-adjointextensionHrel}
Let $0 < \kappa \leq 1$ and $|\gamma| M_{\kappa/2}^2 < 1$. Then, 
\begin{align}\label{eq_definitionHrelF}
H^{\mathrm{rel}}_F 
:= 
\left(
\begin{array}{cc}
\mathrm{id} & \mathbf{0} \\
P_- \gamma V P_+ A^{- 1} \quad & \mathrm{id}
\end{array}\right) 
\left(
\begin{array}{cc}
A & \mathbf{0} \\
\mathbf{0} & S_F
\end{array}
\right) 
\left(
\begin{array}{cc}
\mathrm{id} & \quad \overline{A^{-1} P_+ \gamma V P_-} \\
\mathbf{0} & \mathrm{id}
\end{array}
\right) 
- P_+ B P_+ 
\end{align}
with domain
\begin{align}\label{eq_definitiondomainHrelF}
\D (H^{\mathrm{rel}}_F) 
= 
\left\{ 
\left.
\left(
\begin{array}{c}
f \\
g
\end{array}
\right) 
\in 
\Hrel_+ \oplus \Hrel_- 
\right|\ 
\begin{array}{c}
f + \overline{A^{-1} P_+ \gamma V P_-} g \in \mathcal{D}_+, \\
g \in \D (S_F) 
\end{array}
\right\} 
\end{align}
defines a self-adjoint extension of $H^{\mathrm{rel}}$. 
\end{theorem}

\begin{proof}
As $P_+ B P_+$ is bounded and symmetric, it suffices to show self-adjointness of $H^{\mathrm{rel}}_F + P_+ B P_+$, for which we introduce the short-hand notation 
\begin{align}\label{eq_def_operators_RST}
\mathcal{R} \, \mathcal{S}_F \, \mathcal{T}
\equiv
\left(
\begin{array}{cc}
\mathrm{id} & \mathbf{0} \\
P_- \gamma V P_+ A^{- 1} \quad & \mathrm{id}
\end{array}\right) 
\left(
\begin{array}{cc}
A & \mathbf{0} \\
\mathbf{0} & S_F
\end{array}
\right) 
\left(
\begin{array}{cc}
\mathrm{id} & \quad \overline{A^{-1} P_+ \gamma V P_-} \\
\mathbf{0} & \mathrm{id}
\end{array}
\right) .
\end{align}
where the three operators $\mathcal{R}$, $\mathcal{S}_F$, and $\mathcal{T}$ correspond to the three matrix operators on the right, respectively. We have already shown in the proof of Lemma~\ref{lemma_frobschurofHrel} that the assumptions~(A\ref{assumptions_matrixop_1})-(A\ref{assumptions_matrixop_6}) from Appendix~\ref{section_matrixoperators} are met. Hence, by Lemma~\ref{lemma_boundedmatrixop} from Appendix~\ref{section_matrixoperators}, $\mathcal{R}$ and $\mathcal{T}$ are bounded and boundedly invertible as well as $\mathcal{R}^* = \mathcal{T}$ and $\mathcal{T}^* = \mathcal{R}$ hold. Moreover, self-adjointness of $A$ (Lemma~\ref{lemma_Asa0resolvent}) and of $S_F$ (Lemma~\ref{lemma_schursym}) implies self-adjointness of $\mathcal{S}_F$. It is well-known that operator products with these properties are self-adjoint (see, e.g., \cite[Lemma~10.18]{schmuedgen}), and so, self-adjointness of $H^{\mathrm{rel}}_F$ follows. That $H^{\mathrm{rel}}_F$ extends $H^{\mathrm{rel}}$ follows from 
$\D (H^{\mathrm{rel}}) \subseteq \D (H^{\mathrm{rel}}_F)$ 
and $S \subseteq S_F$ (Lemma~\ref{lemma_schursym}). 
\end{proof}

\begin{theorem}[Claim~\ref{thm:main1_b}) of Theorem~\ref{thm:main1}]\label{thm:sa}
Let $0 < \kappa \leq 1$ and $|\gamma| M_{\kappa/2}^2 < 1$. Using that $\hat{\vec{P}}$ also denotes the operator of multiplication with $\vec{P} \in \R^3$, we define 
\begin{align} 
H_F := 
\PMpl \otimes \mathrm{id}_{L^2 (\dd r)} + 
\mathrm{id}_{L^2 (\dd R)} \otimes H^{\mathrm{rel}}_F
\end{align}
with domain 
\begin{align}
\D (H_F) =& 
\left\{ \left. \vphantom{\int_{\R^3}}
f \in L^2 (\R^3, \dd P; \Hrel) \right|\ f (\vec{P}) \in \D (H^\mathrm{rel}_F) \; \mathrm{for} \; \mathrm{almost}\; \mathrm{all} \; \vec{P}\in \R^3  
\right. \nonumber \\
&\qquad \left. \mathrm{and} \; \int_{\R^3}  \left\| \vec{P} \cdot \mathbf{M}^+ + H^\mathrm{rel}_F \right\|^2 \dd P < \infty \right\} .
\end{align}
Then, upon defining the unitary operator $\mathcal{W} := \F_{\vec{R}} U$, where $U$ is the coordinate transform \eqref{eq_transformationtorelcomcoord} and $\F_{\vec{R}}$ is the Fourier transform with respect to the center-of-mass coordinate, it holds that 
\begin{align}
\tilde{H}_{\mathrm{DC}} := \mathcal{W}^{-1} H_F \mathcal{W} + V_{\mathrm{ext}} + \beta m_1 \otimes \mathbf{1}_4 + \mathbf{1}_4 \otimes \beta m_2
\end{align}
with domain 
\begin{align}\label{eq_def_domain_Htilde_F}
\D (\tilde{H}_{\mathrm{DC}}) = \mathcal{W}^{-1} \D (H_F) 
\end{align}
defines a self-adjoint extension of $\HDC$. 
\end{theorem}

\begin{proof}
First, we note that $\beta m_1 \otimes \mathbf{1}_4 + \mathbf{1}_4 \otimes \beta m_2$ as well as $V_{\mathrm{ext}}$ are symmetric and bounded. Therefore, they can be added by means of a bounded perturbation. Moreover, the coordinate transformation $U$ as well as the Fourier transform of the center-of-mass coordinate are unitary. Hence, self-adjointness of $H_F$ implies self-adjointness of $\tilde{H}_{\mathrm{DC}}$.  

In order to prove self-adjointness of $H_F = \PMpl \otimes \mathrm{id} + \mathrm{id} \otimes H_F^{\mathrm{rel}}$, we employ the method of direct fiber integrals. We define
$H_F (\vec{P}) := \vec{P} \cdot \mathbf{M}^+ + H^\mathrm{rel}_F $
for a fixed $\vec{P} \in \R^3$ and so 
$\{ H_F (\vec{P}) \}_{\vec{P}\in\R^3}$ 
is a family of self-adjoint operators with common domain 
$\D (H^\mathrm{rel}_F)$ 
in the underlying Hilbert space $\Hrel$ by Theorem~\ref{thm_self-adjointextensionHrel}. The map 
$\vec{P} \mapsto H_F (\vec{P})$
from $\R^3$ into the self-adjoint operators on $\Hrel$ is measurable as for all $f,g \in \Hrel$ the map $\vec{P} \mapsto \left\langle f, (H_F (\vec{P}) + i)^{-1} g \right\rangle$ is continuous in $\vec{P}$:
\begin{align}
&\left| 
\left\langle f, \left( H_F (\vec{P}) + i \right)^{-1} g \right\rangle 
- 
\left\langle f, \left( H_F (\vec{P}') + i \right)^{-1} g \right\rangle 
\right| \nonumber\\
& \overset{(\ast)}{=} 
\left| 
\left\langle f, \left( H_F (\vec{P}) + i \right)^{-1} 
\left( H_F (\vec{P}') - H_F (\vec{P}) \right) 
\left( H_F (\vec{P}') + i \right)^{- 1} g 
\right\rangle 
\right| \nonumber\\
& = 
\left|
\left\langle f, \left( H_F (\vec{P}) + i \right)^{- 1} 
(\vec{P}' -\vec{P}) \cdot \mathbf{M}^+
\left( H_F (\vec{P}') + i \right)^{- 1} g 
\right\rangle 
\right| \nonumber\\
&\leq 	\norm{f} \norm{\left( H_F (\vec{P}') + i \right)^{- 1}}^2 \sum_{k = 1}^3 \norm{\mathrm{M}^+_k} \Abs{P_k' - P_k} \norm{g} \xrightarrow{\vec{P} \rightarrow \vec{P}'} 0
\end{align}
since $\| (H_F (\vec{P}) + i )^{-1} \| \leq 1$. In $(\ast)$, we used the second resolvent identity. 

Hence, due to \cite[Theorem~XIII.85]{RS4}, the direct fiber integral 
\begin{align}
H'
=
\int_{\R^3}^\oplus
H_F (\vec{P})
\dd P
\end{align}
with domain 
\begin{align}
\D (H') =& 
\left\{ \left. \vphantom{\int_{\R^3}}
f \in L^2 (\R^3, \dd P; \Hrel) \right|\ f (\vec{P}) \in \D (H^\mathrm{rel}_F) \; \mathrm{for} \; \mathrm{almost}\; \mathrm{all} \; \vec{P}\in \R^3  
\right. \nonumber \\
&\qquad \left. \mathrm{and} \; \int_{\R^3}  \norm{H_F (\vec{P}) f (\vec{P})}_{\Hrel}^2 \dd P < \infty \right\}  
\end{align}
defines a self-adjoint operator in the Hilbert space 
$L^2 (\R^3, \dd P ; \Hrel)$
which acts as $(H' f) (\vec{P}) = H_F (\vec{P}) f (\vec{P})$ for almost all $\vec{P} \in \R^3$. This is precisely the action of $H_F$. 
Since 
$\D_0 = H^1 (\R^3, \dd P) \otimes \C^{16} \otimes H^1 (\R^3, \dd r) $
is contained in $\D (H_F)$ and $H^{\mathrm{rel}}_F$ extends $H^{\mathrm{rel}}$, we obtain  
$H_F \upharpoonright \D_0 = \mathcal{W} \HDC \mathcal{W}^{-1}$. 
This proves the theorem. 
\end{proof}

\subsection{Proof of Theorem~\ref{thm_coreV1/2}} \label{section_prooftechnicalstuff}

In order to prove Theorem~\ref{thm_coreV1/2}, we will need to control the integral kernel of the operator $\tau$ from Definition~\ref{def_tau_projections}. There, $\tau$ was defined as Fourier multiplier, i.e., for all $f \in \Hrel$ and for almost all $\vec{r} \in \R^3$, we had 
\begin{align}
(\tau f) (\vec{r}) 
=
\lim_{M \rightarrow \infty} \int_{| \vec{p} | \leq M} e^{2 \pi i\vec{r} \cdot \vec{p}} \, \tau (\vec{p}) \hat{f} (\vec{p}) \dd p .
\end{align}
All entries of the matrix $\tau (\vec{p})$ are fractions. By expanding the product in \eqref{eq_tauofp}, one finds that the denominator of all of them is $p^2$, whereas each numerator is given by one of the following expressions: 
\begin{align}
p_3^2 \, , \quad (p_1 \pm i p_2) p_3 \, , \quad (p_1 \pm i p_2)^2 \, , \quad (p_1 - i p_2) (p_1 + i p_2) .
\end{align}
In order to control $\tau$, it therefore suffices to control the corresponding position space integral kernels $K_{ij} (\vec{y})$ of the following operator $T_{ij}$, $i,j = 1,2,3$, which is defined for all $f \in L^2 (\R^3, \dd r)$ and for almost all $\vec{r} \in \R^3$ by  
\begin{align}\label{eq_def_Tij}
(T_{ij} f) (\vec{r}) 
:=
\lim_{M \rightarrow \infty} \int_{| \vec{p} | \leq M} e^{2 \pi i\vec{r} \cdot \vec{p}} \, \frac{p_i\, p_j}{p^2} \hat{f} (\vec{p}) \dd p ,
\end{align}
similar to the definition of $\tau$, where the limit exists in the $L^2$-sense. The factor $1/p^2$ suggests that we will need the following theorem, which we cite from \cite{liebloss} in the case of $\R^3$. 

\begin{theorem}\label{thm_hardylittlewoodsobolev}
Let $\alpha \in \R^+$, and define $c_\alpha := \pi^{-\alpha / 2} \Gamma (\alpha / 2)$. 
\begin{enumerate}[a)]
	\item \label{thm_hardylittlewoodsobolev_a}
	Let $f$ be a function in $C_c^\infty (\R^3)$ and let $0 < \alpha < 3$. Then, 
	\begin{align}
	c_\alpha 
	\int_{\R^3} 
	e^{2 \pi i \vec{r} \cdot \vec{p}} \frac{1}{p^\alpha} \hat{f} (\vec{p})
	\dd p
	=
	c_{3-\alpha} 
	\int_{\R^3}
	\frac{1}{|\vec{r}-\vec{y}|^{3-\alpha}} f (\vec{y})
	\dd y .
	\end{align}
	
	\item \label{thm_hardylittlewoodsobolev_b}
	If $0 < \alpha < 3/2$ and if $f \in L^p (\R^3, \mathrm{d}^3 r)$ with $p = 6/(3+2 \alpha)$, then $\hat{f}$ exists in the sense of the Hausdorff-Young inequality (see \cite[Theorem~5.7]{liebloss}). Moreover, defining the function $g$ for almost all $\vec{r} \in \R^3$ by 
	\begin{align}
	g (\vec{r}) 
	= 
	c_{3-\alpha} 
	\int_{\R^3}
	\frac{1}{|\vec{r}-\vec{y}|^{3-\alpha}} f (\vec{y})
	\dd y ,
	\end{align}
	we have $g \in L^2 (\R^3, \mathrm{d}^3 r)$. Furthermore, for almost all $\vec{p} \in \R^3$, we obtain 
	\begin{align}
	c_\alpha |\vec{p}|^{-\alpha} \hat{f} (\vec{p}) = \hat{g} (\vec{p}) . 
	\end{align}
\end{enumerate}
\end{theorem}

\begin{proof}
\begin{enumerate}[a)]
	\item See \cite[Theorem~5.9]{liebloss}. 
	\item See \cite[Corollary~5.10]{liebloss}. 
	\qedhere
\end{enumerate}
\end{proof}

In order to compute the integral kernels $K_{ij} (\vec{y})$, we will also need a result from the theory of Calder\'{o}n-Zygmund singular integrals. For later reference, we state it in the following in the case of $\R^3$. 

\begin{theorem}\label{thm_calderonzygmund}
If $K(\vec{y})$ is a homogeneous function of degree $-3$, i.e., 
$K (\lambda \vec{y}) = \lambda^{-3} K (\vec{y})$ 
for almost all $\vec{y} \in \R^3$ and for all $\lambda > 0$, 
and if $K (\vec{y})$ has in addition the following properties 
\begin{enumerate}[a)]
	\item 
	$\displaystyle
	\int_{S^2} K (\vec{y}) \,\mathrm{d} \sigma (\vec{y}) = 0 $, 
	where $S^2$ denotes the surface of the unit ball and $\mathrm{d} \sigma$ the corresponding surface measure, and \\
	
	\item 
	$\displaystyle 
	\int_{S^2} \left| K (\vec{y}) + K (-\vec{y}) \right| \log^+ \left| K (\vec{y}) + K (-\vec{y}) \right| \, \mathrm{d} \sigma (\vec{y}) < \infty $,
	where $\log^+$ denotes\\ 
	
	the positive part of the logarithm, 
\end{enumerate}
then, if $f \in L^p (\R^3, \dd r)$, $1 < p < \infty$, the limit 
\begin{align}\label{eq_truncatedsingint}
\lim_{\varepsilon \rightarrow 0}
\int_{|\vec{y}|>\varepsilon} 
K(\vec{y}) f (\vec{r} - \vec{y}) 
\dd y 
\end{align} 
exists in $L^p$-sense and pointwise for almost all $\vec{r} \in \R^3$. Furthermore, there exists a constant $A>0$, depending on $p$ and $K$ only, such that 
\begin{align}
\left\| 
\sup_{\varepsilon > 0} 
\left| \int_{| \vec{y} | > \varepsilon}
K (\vec{y}) f (\cdot -\vec{y}) \dd y 
\right|
\right\| 
\leq 
A \left\|	f \right\| .
\end{align}
\end{theorem}

\begin{proof}
See \cite[Theorem~1]{calderonzygmund1956}. 
\end{proof}

\begin{lemma}\label{lemma_kerneltau}
For all $f \in L^2 (\R^3, \dd r)$, all $i,j = 1,2,3$, and almost all $\vec{r} \in \R^3$, the action of $T_{ij}$, defined in \eqref{eq_def_Tij}, can be given by the formula   
\begin{align}
\left( T_{ij} f \right) (\vec{r})
=
\frac{\delta_{ij}}{3}
f (\vec{r}) 
-
\frac{1}{4\pi}
\lim_{\varepsilon \rightarrow 0} 
\int_{| \vec{y} | > \varepsilon}
K_{ij} (\vec{y}) f (\vec{r}-\vec{y}) 
\dd y ,
\end{align}
where the Kronecker delta $\delta_{ij}$ is $1$ if $i=j$ and $0$ else. The integral kernels $K_{ij} (\vec{y})$ of $T_{ij}$ are then for all $i,j = 1,2,3$ and all $|\vec{y}| > 0$ given by 
\begin{align}
K_{ij} (\vec{y}) 
:= 
\frac{\partial^2}{\partial y_i \partial y_j} \frac{1}{| \vec{y} |} 
=
\left\{
\begin{aligned}
&\frac{3 y_i y_j}{| \vec{y} |^5} \, ,
\quad  i \neq j \\
\frac{3 y^2_i}{| \vec{y} |^5} &- \frac{1}{| \vec{y} |^3} \, , \quad i = j. 
\end{aligned}
\right.
\end{align}
\end{lemma}

\begin{proof}
First, we let $f \in C_c^\infty (\R^3)$. Using \eqref{eq_def_Tij}, we calculate for all $\vec{r} \in \R^3$ 
\begin{align}
-4 \pi \, & (T_{ij} f) (\vec{r}) 
=
-4 \pi
\int_{\R^3} 
e^{2 \pi i\vec{r} \cdot \vec{p}} \, \frac{p_i\, p_j}{p^2} \hat{f} (\vec{p}) 
\dd p \nonumber \\
& =
\frac{1}{\pi}
\int_{\R^3} 
\left(
\frac{\partial^2}{\partial r_i\partial r_j} e^{2 \pi i\vec{r} \cdot \vec{p}} 
\right)
\frac{1}{p^2} \hat{f} (\vec{p}) 
\dd p \nonumber \\
& \overset{\mathrm{(i)}}{=} 
\frac{1}{\pi}
\frac{\partial^2}{\partial r_i\partial r_j}
\int_{\R^3} 
e^{2 \pi i\vec{r} \cdot \vec{p}} \frac{1}{p^2} \hat{f} (\vec{p}) 
\dd p 
\overset{\mathrm{(ii)}}{=} 
\frac{\partial^2}{\partial r_i\partial r_j}
\int_{\R^3} 
\frac{1}{| \vec{y} |} f (\vec{r} - \vec{y}) 
\dd y \nonumber \\
& \overset{\mathrm{(iii)}}{=} 
\int_{\R^3} 
\frac{1}{| \vec{y} |}
\frac{\partial^2 f (\vec{r}-\vec{y})}{\partial r_i \partial r_j} 
\dd y 
=
\int_{\R^3} 
\frac{1}{| \vec{y} |}
(-1)^2 \frac{\partial^2 f (\vec{r}-\vec{y})}{\partial y_i \partial y_j} 
\dd y \label{eq_computationofTij}
\end{align}
where we used dominated convergence to commute the derivatives with the integral in (i),  Theorem~\ref{thm_hardylittlewoodsobolev}\ref{thm_hardylittlewoodsobolev_a}) for the Fourier integral in (ii), and the calculus of distributional convolutions and derivatives in (iii) (see \cite[Lemma~6.8]{liebloss}). It applies as $|\cdot |^{-1} \in L^1_{\mathrm{loc}} (\R^3)$ as well as $f \in C_c^\infty (\R^3)$. In order to continue with integration by parts, we change the domain of integration to the set 
$B = \{ \varepsilon < |\vec{y}| < R \}$, where $\varepsilon > 0$ is fixed and $R>0$ is chosen sufficiently large such that $\mathrm{supp} f (\vec{r} - \cdot) \subset B_R(0)$. As $f \in C_c^\infty (\R^3)$, this is possible for all fixed $\vec{r} \in \R^3$. 
We obtain
\begin{subequations}\label{eq_integrationbyparts}
	\begin{align}
	\int_{B} 
	&\frac{1}{| \vec{y} |}
	\frac{\partial^2 f (\vec{r}-\vec{y})}{\partial y_i \partial y_j} 
	\dd y \nonumber \\
	={}&
	-\int_{B} 
	\left(\frac{\partial}{\partial y_i}\frac{1}{| \vec{y} |}\right)
	\left(\frac{\partial f (\vec{r}-\vec{y})}{\partial y_j}\right)
	\dd y 
	+
	\int_{|\vec{y}|=\varepsilon}
	\frac{1}{| \vec{y} |}
	\frac{\partial f (\vec{r}-\vec{y})}{\partial y_j}
	\, \nu^i\mathrm{d}\sigma (\vec{y})
	\nonumber \\
	={}&
	\int_{B} 
	\left(\frac{\partial^2}{\partial y_j\partial y_i}\frac{1}{| \vec{y} |}\right)
	f (\vec{r}-\vec{y})
	\dd y 
	-
	\int_{|\vec{y}|=\varepsilon}
	\left(\frac{\partial}{\partial y_i}\frac{1}{| \vec{y} |}\right)
	f (\vec{r}-\vec{y})
	\, \nu^j\mathrm{d}\sigma (\vec{y})
	\label{eq_integrationbyparts_a} \\
	&{}+
	\int_{|\vec{y}|=\varepsilon}
	\frac{1}{| \vec{y} |}
	\frac{\partial f (\vec{r}-\vec{y})}{\partial y_j}
	\, \nu^i\mathrm{d}\sigma (\vec{y}) , \label{eq_integrationbyparts_b}
	\end{align}
\end{subequations}
where $\nu^k$ denotes the $k$-th component of $\vec{\nu}$, the unit outward normal to $B$. Due to our choice of $R$, the boundary terms at $|\vec{y}| = R$ vanish. 

In order to have equality of lines~\eqref{eq_computationofTij} and \eqref{eq_integrationbyparts}, we need to take the pointwise limits $\varepsilon\rightarrow 0$ and $R \rightarrow \infty$ in line~\eqref{eq_integrationbyparts}. Existence of the latter follows from the compact support of $f$. For the former, we note that the derivative of $f\in C_c^\infty(\R^3)$ is uniformly bounded by a constant $C>0$ independent of $\varepsilon$. We thus get for line~\eqref{eq_integrationbyparts_b} 
\begin{align}
\left|
\int_{|\vec{y}|=\varepsilon}
\frac{1}{| \vec{y} |}
\frac{\partial f (\vec{r}-\vec{y})}{\partial y_j}
\, \nu^i\mathrm{d}\sigma (\vec{y})
\right|
&\leq
\int_{S^2}
\frac{1}{\varepsilon}
\left|\frac{\partial f (\vec{r}-\vec{y})}{\partial y_j}\right|
\,\varepsilon^2\mathrm{d}\Omega \nonumber \\
&\leq
C \varepsilon
\int_{S^2}
\,\mathrm{d}\Omega 
\xrightarrow{\varepsilon\rightarrow 0} 0 .
\end{align}
In the last summand of line~\eqref{eq_integrationbyparts_a}, we use the parametrization $\Phi$ of the boundary at $|\vec{y}|=\varepsilon$ with spherical coordinates and note that $\Phi (\nabla |\vec{y}|^{-1}) = \vec{\nu}/\varepsilon^2$. Thus, 
\begin{align}
\lim_{\varepsilon \rightarrow 0}
&\int_{|\vec{y}|=\varepsilon}
\left(\frac{\partial}{\partial y_i}\frac{1}{| \vec{y} |}\right)
f (\vec{r}-\vec{y})
\, \nu^j\mathrm{d}\sigma (\vec{y}) \nonumber \\
&= 
\lim_{\varepsilon \rightarrow 0}
\int_{S^2}
\frac{\nu^i}{\varepsilon^2} f (\vec{r}-\vec{y}(\varepsilon,\theta, \phi))
\,\nu^j\,\varepsilon^2\mathrm{d}\Omega \nonumber \\
&= 
\int_{S^2}
\nu^i  \nu^j \lim_{\varepsilon \rightarrow 0} f (\vec{r}-\vec{y}(\varepsilon,\theta, \phi))
\,\mathrm{d}\Omega 
=
f (\vec{r})
\int_{S^2}
\nu^i  \nu^j 
\,\mathrm{d}\Omega 
\end{align} 
follows from dominated convergence and continuity of $f$. 
Now, we need to distinguish the cases $i \neq j$ and $i = j$. We obtain by direct computation 
\begin{align}
f (\vec{r}) \int_{S^2} \nu^i \nu^j \,\mathrm{d}\Omega
=
\left\{
\begin{aligned}
&0 \, , \quad i \neq j , \\
\frac{4\pi}{3} &f (\vec{r}) \, , \quad i = j .
\end{aligned}
\right.
\end{align}
It remains to investigate the limit $\varepsilon \rightarrow 0$ for the first summand of \eqref{eq_integrationbyparts_a}. This can be achieved with the help of Theorem~\ref{thm_calderonzygmund} which, however, requires an analysis of the singular integral kernel
$\partial^2/\partial y_i\partial y_j |\vec{y}|^{-1}$. 
By explicit computation, we obtain for $|\vec{y}| > 0$ 
\begin{align}
K_{ij} (\vec{y}) 
= 
\frac{\partial^2}{\partial y_i \partial y_j} \frac{1}{| \vec{y} |} 
=
\left\{
\begin{aligned}
&\frac{3 y_i y_j}{| \vec{y} |^5} \, ,
\quad  i \neq j \\
\frac{3 y^2_i}{| \vec{y} |^5} &- \frac{1}{| \vec{y} |^3} \, , \quad i = j. 
\end{aligned}
\right.
\end{align}
We check the conditions of Theorem~\ref{thm_calderonzygmund} and observe first that $K_{ij} (\vec{y})$ is homogeneous of degree $-3$, i.e., 
$K_{ij} (\lambda \vec{y}) = \lambda^{-3} K_{ij} (\vec{y})$ 
for almost all $\vec{y} \in \R^3$ and for all $\lambda > 0$. Moreover, $K_{ij} (\vec{y})$ exhibits the following crucial property for all $i,j = 1, 2, 3$:
\begin{align}
\int_{S^2} K_{ij} (\vec{y}) \,\mathrm{d} \sigma (\vec{y}) = 0 .
\end{align}
Furthermore, since $K_{ij} (-\vec{y}) = K_{ij} (\vec{y})$, it suffices to estimate for all $i,j = 1, 2, 3$
\begin{align}
\int_{S^2} \left| K_{ij} (\vec{y}) \right| \log^+ \left| K_{ij} (\vec{y}) \right| \, \mathrm{d} \sigma (\vec{y}) 
\leq
\int_{S^2} \left| K_{ij} (\vec{y}) \right|^2 \, \mathrm{d} \sigma (\vec{y})
<
\infty
\end{align}
where $\log^+$ denotes the positive part of the logarithm. That the last integral is indeed finite can be seen by explicit calculation. 

Hence, all conditions of Theorem~\ref{thm_calderonzygmund} are met and we can conclude that the limit 
\begin{align}
\lim_{\varepsilon \rightarrow 0} 
\int_{| \vec{y} | > \varepsilon}
K_{ij} (\vec{y}) f (\vec{r}-\vec{y}) 
\dd y
\end{align}
exists pointwise for almost all $\vec{r} \in \R^3$ as well as in $L^2 (\R^3, \dd r)$. We have thus established for all $\vec{r} \in \R^3$ and $f\in C_c^\infty (\R^3)$ 
\begin{align}
\left( T_{ij} f \right) (\vec{r})
=
\frac{\delta_{ij}}{3}
f (\vec{r}) 
-
\frac{1}{4\pi}
\lim_{\varepsilon \rightarrow 0} 
\int_{| \vec{y} | > \varepsilon}
K_{ij} (\vec{y}) f (\vec{r}-\vec{y}) 
\dd y ,
\end{align}
where $\delta_{ij}$ denotes the Kronecker delta. 

In order to extend the formula for the action of $T_{ij}$ to all of $L^2 (\R^3, \dd r)$, we note that Theorem~\ref{thm_calderonzygmund} also guarantees the existence of a constant $A>0$ such that for all $f \in L^2 (\R^3, \dd r)$
\begin{align}
\left\| 
\sup_{\varepsilon > 0} 
\left| \int_{| \vec{y} | > \varepsilon}
K_{ij} (\vec{y}) f (\cdot -\vec{y}) \dd y 
\right|
\right\| 
\leq 
A \left\|	f \right\| .
\end{align}
Choose now any $f \in L^2 (\R^3, \dd r)$ and fix a sequence $(f_n)_{n \in \N} \subset C_c^\infty (\R^3)$ such that 
$\left\| f - f_n \right\| \xrightarrow{n \rightarrow \infty} 0$. 
Then, 
\begin{align}
\left\|
\right. & \left. T_{ij} \left( f - f_n \right)
\right\| \nonumber \\
& \leq 
\frac{\delta_{ij}}{3}
\left\| f - f_n \right\|
+
\frac{1}{4\pi}
\left\|
\sup_{\varepsilon > 0} 
\left|
\int_{| \vec{y} | > \varepsilon} 
K_{ij} (\vec{y})  \left( f (\cdot -\vec{y}) - f_n (\cdot -\vec{y}) \right)
\dd y
\right|
\right\| \nonumber \\
& \leq  
\left( \frac{\delta_{ij}}{3} + \frac{A}{4\pi} \right)
\left\| f - f_n \right\| 
\xrightarrow{n \rightarrow \infty} 0 .
\end{align}
This proves the statement of the lemma. 
\end{proof}

\setcounter{theorem}{2}
\begin{theorem}
Let $0 < \kappa \leq 1$. For every 
$f \in \D (V^{1/2}) \cap \Hrel_-$, 
there exists a sequence 
$(f_n)_{n \in \N} \subset \D (V)$
such that
\begin{align}
\left\| 
f - P_- f_n 
\right\| 
+ 
\left\| 
V^{1/2} P_- (f - f_n) 
\right\| 
\xrightarrow{n \rightarrow \infty} 0 . 
\end{align}
\end{theorem}
\setcounter{theorem}{7}

\begin{proof}
Let $f \in \mathcal{D} (V^{1/2}) \cap \Hrel_-$. First, we define the Gaussian $\chi_n$ for all $n \in \N$ by 
\begin{align}
\chi_n (\vec{r}):= e^{- 2^n \pi | \vec{r} |^2} ,
\end{align}
and with it the Gaussian cut-off $1 - \chi_n$ for $f$ by $f_n := (1 - \chi_n) f$. By construction, $(f_n)_{n \in \N} \subset \D (V)$ for all $n \in \N$ since we have $(1 - \chi_n (\vec{r})) |\vec{r}|^{-\kappa} \rightarrow 0$ as $|\vec{r}|\rightarrow 0$ for all $n \in \N$. Dominated convergence implies 
\begin{align}
\left\| f - f_n \right\|^2 
=
\left\| f - (1 - \chi_n) f \right\|^2 
= 
\int_{\R^3} 
\left| \chi_n (\vec{r}) \right|^2  | f (\vec{r}) |^2 
\dd r \xrightarrow{n \rightarrow \infty} 0
\end{align}
since 
$| \chi_n (\vec{r}) |^2 \leq 1$ 
and 
$| \chi_n (\vec{r}) |^2 \xrightarrow{n \rightarrow \infty} 0$ 
for almost all $\vec{r} \in \R^3$, and therefore, 
\begin{eqnarray}
\| f - P_- f_n \| 
& = & 
\| P_- f - P_- f_n \| 
\leq 
\| P_- \|  \| f - f_n \| \xrightarrow{n \rightarrow \infty} 0 .
\end{eqnarray}
In order to show 
$\left\| V^{1 / 2} P_- (f - f_n) \right\| \xrightarrow{n \rightarrow \infty} 0$, 
we estimate
\begin{eqnarray}\label{eq_calc_core_thm}
\left\| 
V^{1/2} P_- (f - f_n) 
\right\|
& \leq & 
\left\| 
P_- V^{1/2} \chi_n f 
\right\| 
+
\left\| 
\left[V^{1/2}, P_-\right] \chi_n f 
\right\| 
\nonumber \\ 
& \leq &
\left\| P_- \right\|
\left\| 
V^{1/2} \chi_n f 
\right\| 
+
\left\| 
\left[V^{1/2}, \tau \right] \chi_n f 
\right\| ,
\end{eqnarray}
where we used $P_- = 1/2 \, (\mathrm{id} - \tau)$. Hence, using $f \in \D (V^{1 / 2})$, we get for the first summand of \eqref{eq_calc_core_thm} by dominated convergence as above
\begin{align}
\left\| 
V^{1/2} \chi_n f 
\right\| 
\xrightarrow{n \rightarrow \infty} 0 . 
\end{align}
In order to treat the second summand of \eqref{eq_calc_core_thm}, we note the following. In the paragraph leading to the definition of $T_{ij}$ in line~\eqref{eq_def_Tij}, we saw that in order to control $\tau$, it suffices to control $T_{ij}$. Thus, denoting the spinor components of $f$ by $f^k$, $k = 1,2,\ldots, 16$, it remains to show 
\begin{align}
\left\| 
\left[ 
V^{1/2} , T_{ij} 
\right]  
\chi_n f^k
\right\|
\xrightarrow{n \rightarrow \infty} 0 
\end{align}
for all $i,j=1,2,3$ and $k=1,2,\ldots, 16$.  

In the following, we will collect all numerical factors, which are independent of $\varepsilon$, $\vec{r}$, and $\vec{y}$, by the same symbol $C$. The exact value of $C$ might therefore change from one line to the next. 

Using that $V^{1/2}$ is multiplication with $| \cdot|^{-\kappa/2}$ and thus commutes with that summand of $T_{ij}$ which contains the Kronecker delta $\delta_{ij}$, we estimate for almost all $\vec{r} \in \R$ and for all $i,j,k$ 
\begin{align}
&\left| 
\left(
\left[ 
| \cdot|^{-\kappa/2}, T_{ij} 
\right] 
\chi_n f^k 
\right)
(\vec{r}) 
\right| \nonumber \\
& = 
C
\left| 
\lim_{\varepsilon \rightarrow 0}  
\int_{| \vec{y} | > \varepsilon} 
K_{ij} (\vec{y}) 
\left( 
\frac{1}{| \vec{r} |^{\kappa / 2}} - \frac{1}{| \vec{r}-\vec{y} |^{\kappa / 2}} 
\right)
\chi_n (\vec{r}-\vec{y}) f^k (\vec{r}-\vec{y})
\dd y 
\right| \nonumber \\
& \leq 
C 
\lim_{\varepsilon \rightarrow 0} 
\int_{| \vec{y} | > \varepsilon} 
\frac{1}{| \vec{y} |^3} 
\left| 
\frac{1}{| \vec{r} |^{\kappa / 2}} - \frac{1}{| \vec{r}-\vec{y} |^{\kappa / 2}} 
\right| 
\chi_n (\vec{r}-\vec{y}) \left| f^k (\vec{r}-\vec{y}) \right| 
\dd y \nonumber \\
& = 
C 
\lim_{\varepsilon \rightarrow 0} 
\int_{| \vec{y} | > \varepsilon} 
\frac{1}{| \vec{y} |^3} 
\left|
\frac{| \vec{r}-\vec{y} |^{\kappa / 2} - | \vec{r} |^{\kappa / 2}}{| \vec{r}-\vec{y} |^{\kappa / 2} | \vec{r} |^{\kappa / 2}} 
\right|
\chi_n (\vec{r}-\vec{y}) \left| f^k (\vec{r}-\vec{y}) \right| 
\dd y \nonumber \\
& \overset{(\ast)}{\leq} 
C \lim_{\varepsilon \rightarrow 0}  
\int_{| \vec{y} | > \varepsilon} 
\frac{1}{| \vec{y} |^3}  \frac{| \vec{y} |^{\kappa / 2}}{| \vec{r}-\vec{y} |^{\kappa / 2} | \vec{r} |^{\kappa / 2}} 
\chi_n (\vec{r}-\vec{y}) \left| f^k (\vec{r}-\vec{y}) \right| 
\dd y \nonumber \\
& = 
C \frac{1}{| \vec{r} |^{\kappa / 2}}
\lim_{\varepsilon \rightarrow 0}  
\int_{| \vec{y} | > \varepsilon}
\frac{1}{| \vec{y} |^{3 - \kappa / 2}}  \frac{\chi_n (\vec{r}-\vec{y}) \left| f^k (\vec{r}-\vec{y}) \right|}{|\vec{r}-\vec{y} |^{\kappa / 2}} 
\dd y \nonumber \\
& = 
C \frac{1}{| \vec{r} |^{\kappa / 2}} 
\int_{\R^3} \frac{1}{| \vec{y} |^{3 - \kappa / 2}} 
\frac{\chi_n (\vec{r}-\vec{y}) \left| f^k (\vec{r}-\vec{y}) \right|}{| \vec{r}-\vec{y} |^{\kappa / 2}}
\dd y
\label{eq_commutator_estimate}
\end{align}
where we used Theorem~\ref{thm_hardylittlewoodsobolev}\ref{thm_hardylittlewoodsobolev_b}) together with dominated convergence in the last step and Lemma~\ref{lemma_modtriangle} (Appendix~\ref{section_auxlemmas}) in $(\ast)$. We define $h_n$ for almost all $\vec{r} \in \R^3$ by 
\begin{align}\label{eq_def_hn}
h_n (\vec{r}) 
:=
\int_{\R^3} 
\frac{1}{| \vec{y} |^{3 - \kappa / 2}} \frac{\chi_n (\vec{r} -\vec{y}) \left| f^k (\vec{r} -\vec{y}) \right|}{| \vec{r} -\vec{y} |^{\kappa / 2}} 
\dd y . 
\end{align}
By Theorem~\ref{thm_hardylittlewoodsobolev}\ref{thm_hardylittlewoodsobolev_b}), $h_n \in L^2 (\R^3, \dd r)$ for all $n \in \N$ since we have $| \cdot |^{-\kappa / 2}\chi_n |f^k| \in L^p (\R^3, \dd r)$ for all $p \in [1,2]$ and $n \in \N$ as $\chi_n$ is a Gaussian and $| \cdot |^{-\kappa / 2} |f^k| \in L^2 (\R^3, \dd r)$. Thus, $h_n$ has a Fourier transform, which, also by Theorem~\ref{thm_hardylittlewoodsobolev}\ref{thm_hardylittlewoodsobolev_b}), is for almost all $\vec{p} \in \R^3$ given by
\begin{align}\label{eq_fouriertransform_hn}
\hat{h}_n(\vec{p}) 
= 
\frac{C}{| \vec{p} |^{\kappa/2}} 
\mathcal{F}
\left( \frac{\chi_n \left| f^k \right|}{| \cdot |^{\kappa / 2}} \right) (\vec{p}) 
\end{align}
A further consequence, known as Hardy-Littlewood-Sobolev theorem of fractional integration (see \cite[Chapter~V, Theorem~1]{stein_singularintegrals}), is the inequality 
\begin{align}\label{eq_applicationHLS}
\| h_n \|_2 \leq C \left\| \frac{\chi_n \left| f^k \right|}{| \cdot |^{\kappa / 2}} \right\|_{\frac{6}{3+\kappa}}
\xrightarrow{n\rightarrow\infty} 0, 
\end{align}
where the convergence to zero follows from dominated convergence since $\chi_n \leq \chi_1$ and $\chi_1$ is a Gaussian and thus lies in $L^p(\R^3, \dd r)$ for all $p \geq 1$. Putting everything together then yields 
\begin{eqnarray}
\left\| 
\left[ 
| \cdot|^{-\kappa/2}, T_{ij} 
\right] 
\chi_n f^k 
\right\|
& \overset{(\mathrm{i})}{\leq} & 
C 
\left\| 
| \cdot|^{-\kappa/2}  h_n 
\right\| 
\overset{(\mathrm{ii})}{\leq} 
C 
\left\| 
\left( 
\left|\hat{\vec{p}}\right|^2 + 1
\right)^{\kappa / 4}  h_n
\right\| \nonumber \\
& \leq & 
C
\left(
\left\| 
\left|\hat{\vec{p}}\right|^{\kappa / 2}  h_n
\right\| 
+
\| h_n \|
\right) \nonumber \\
& \overset{\mathrm{(iii)}}{\leq} & 
C 
\left\| 
\frac{\chi_n \left| f^k \right|}{| \cdot |^{\kappa / 2}} 
\right\|_2
+
C
\left\| \frac{\chi_n \left| f^k \right|}{| \cdot |^{\kappa / 2}} \right\|_{\frac{6}{3+\kappa}}
\xrightarrow{n \rightarrow \infty} 0 ,
\end{eqnarray}
where we used 
\eqref{eq_commutator_estimate}, \eqref{eq_def_hn} in (i),
\eqref{eq_modifiedHerbst_1} from the proof of Lemma~\ref{lemma_modifiedHerbst} in (ii), 
\eqref{eq_fouriertransform_hn}, \eqref{eq_applicationHLS} in (iii), 
and dominated convergence and again line~\eqref{eq_applicationHLS} for the convergence $n \rightarrow \infty$. This proves the statement. 
\end{proof}

\subsection{Criterion for distinguished extension}
\label{section_criterionextension}
From a physical point of view, it is very desirable that a self-adjoint extension of a many-body Dirac operator is distinguished by a physical criterion. This criterion should satisfy two requirements which we adopt from \cite{nenciu_sa_dirac}. First, it must have a clear physical meaning. This ensures that the extension is not only an abstract operator but corresponds to the physics the operator is supposed to describe. Secondly, it has to guarantee uniqueness in such a way that there exists only one extension that satisfies this criterion and hence provides a unique unitary time evolution. 

We say that a state $\psi \in \mathcal{D} (H_F)$ has finite potential energy if 
$|E_{\mathrm{pot}} [\psi] | = |\langle f, (V_{\mathrm{ext}}+V_{\mathrm{int}}) f \rangle | < \infty$. 
As physical criterion for a distinguished self-adjoint extension, we choose the one of finite potential energy. This renders it physically meaningful immediately. And as the following theorem shows, it meets also the second requirement as it singles out $H_F$ uniquely. Recall that 
$\D_+ = (\C^{16} \otimes H^1 (\R^3, \dd r)) \cap \Hrel_+ \subset \D (V^{1/2})$. 

\begin{theorem}[Claim~\ref{thm:main1_c}) of Theorem~\ref{thm:main1}]\label{thm_distinguishedextension}
Let $0 < \kappa \leq 1$, $|\gamma| M_{\kappa/2}^2 < 1$, and $\tilde{H}$ with domain $\D (\tilde{H})$ be any self-adjoint extension of $U \HDC U^{-1}$ of the form 
\begin{align}\label{eq_def_Htilde}
\tilde{H} = \PMpl \otimes \mathrm{id} + \mathrm{id} \otimes \tilde{H}^{\mathrm{rel}} + \beta m_1 \otimes \mathbf{1}_4 + \mathbf{1}_4 \otimes \beta m_2 + V_{\mathrm{ext}}
\end{align}
where $\tilde{H}^{\mathrm{rel}}$ is an arbitrary self-adjoint extension of $H^{\mathrm{rel}}$ with domain $\D (\tilde{H}^{\mathrm{rel}})$. Then, $|\langle f, (V_{\mathrm{ext}}+V_{\mathrm{int}}) f \rangle | < \infty$ for all $f \in \D (\tilde{H})$ if and only if $\tilde{H} = H_F$. 
\end{theorem}

\begin{proof} 
As $V_{\mathrm{ext}}$ is bounded (see \eqref{eq_definitionV_ext}), it suffices to consider $V_{\mathrm{int}}$. We first prove $|\langle f, V_{\mathrm{int}} f \rangle | < \infty$ if $f \in \D (H_F)$. Since $V_{\mathrm{int}}$ acts as identity on $\Hcom$, it suffices to show $|\langle f, V_{\mathrm{int}} f \rangle | \leq \langle f, V f \rangle < \infty$ if $f \in \D (H^{\mathrm{rel}}_F)$. 

Let $f = (f_+, f_-)^{\top} \in \D (H^{\mathrm{rel}}_F)$. Then, 
$f_- \in \D(S_F) = \D (V^{1/2}) \cap \D (S^{\ast})$ 
and 
$f_+ + \overline{A^{-1} P_+ \gamma V P_-} f_- \in \D_+$ 
by the definition of $\D (H^{\mathrm{rel}}_F)$ in line~\eqref{eq_definitiondomainHrelF} and Lemma~\ref{lemma_schursym}. 
Hence, 
$\langle f_-, V f_- \rangle < \infty$ follows. 

In order to show $\langle f, V f \rangle < \infty$, it remains to show
$\langle f_+, V f_+ \rangle < \infty$.
We note that 
\begin{eqnarray}
\left\langle f_+, V f_+ \right\rangle^{1/2} 
= 
\left\| V^{1/2} f_+ \right\| 
& \leq & 
\left\| V^{1/2} (f_+ + \overline{A^{-1} P_+ \gamma V P_-} f_-) \right\| \nonumber \\
& & {} + 
\left\| V^{1/2}  \overline{A^{-1} P_+ \gamma V P_-} f_- \right\| .
\end{eqnarray}
Since
$f_+ + \overline{A^{- 1} P_+ \gamma V P_-} f_- \in \D_+ \subset \mathcal{D} (V^{1/2})$, 
we see that $\langle f_+, V f_+ \rangle < \infty$ holds if $\| V^{1 / 2} \overline{A^{- 1} P_+ \gamma V P_-} f_- \| < \infty$. In order to show that, we expoit $\langle f_-, V f_- \rangle < \infty$. 

From Theorem~\ref{thm_coreV1/2} we know that there exists a sequence 
$(f_n)_{n \in \N} \subset \D (V) \cap \Hrel_-$ 
with 
$\| f_- - f_n \| \xrightarrow{n \rightarrow \infty} 0$ 
such that
$\left\| V^{1/2} (f_m - f_n) \right\| \xrightarrow{m, n \rightarrow \infty} 0$.
With its help, we get 
\begin{align}\label{eq_V12ABfinite}
\left\| \vphantom{V^{1/2}} \right. & \left. V^{1/2} \overline{A^{-1} P_+ \gamma V P_-} (f_m - f_n) \right\| 
= 
\left\| V^{1/2} A^{- 1} P_+ \gamma V P_- (f_m - f_n) \right\| \nonumber \\
& =  
\left\| V^{1/2} P_+ \gamma A^{-1} P_+ V^{1/2} V^{1/2} (f_m - f_n) \right\| \nonumber \\
& \leq 
\left\| V^{1/2} P_+ \gamma A^{-1} P_+ V^{1/2} \right\|  \left\| V^{1/2} (f_m - f_n) \right\| \xrightarrow{m, n \rightarrow \infty} 0 
\end{align}
where we used that 
$V^{1/2} P_+ \gamma A^{-1} P_+ V^{1/2}$
is bounded by Lemma~\ref{lemma_boundedoperatorforformboundedness}. Lemma~\ref{lemma_boundedmatrixop} (Appendix~\ref{section_matrixoperators}) provides boundedness of 
$\overline{A^{-1} P_+ V P_-}$, and we thus obtain the convergence
$\| \overline{A^{-1} P_+ \gamma V P_-} (f_m - f_n) \| \xrightarrow{m, n \rightarrow \infty} 0$. 
Together with line~\eqref{eq_V12ABfinite}, this now ensures that the sequence $(\overline{A^{-1} P_+ V P_-}f_n)_{n\in\N}$ is a Cauchy sequence in the graph norm of $V^{1/2}$. As $V^{1/2}$ is closed, this implies $\overline{A^{- 1} P_+ \gamma V P_-} f_- \in \D(V^{1/2})$, i.e., $\| V^{1 / 2} \overline{A^{- 1} P_+ \gamma V P_-} f_- \| < \infty$. 

For the reverse implication, we assume 
$\langle f, V f \rangle < \infty$ for all 
$f \in \mathcal{D} (\tilde{H})$. 
Then, it suffices to show
$\D (\tilde{H}^{\mathrm{rel}}) = \D (H_F^{\mathrm{rel}})$
in order to infer 
$\tilde{H} = H_F$, i.e., 
$\PMpl + \tilde{H}^{\mathrm{rel}} = \PMpl + H_F^{\mathrm{rel}}$. 

We first show 
$\D (\tilde{H}^{\mathrm{rel}}) \subseteq \D (H^{\mathrm{rel}}_F)$, i.e., let $f \in \D (\tilde{H}^{\mathrm{rel}})$. In order to understand better what $f$ might look like, we first note that 
$\D (\tilde{H}^{\mathrm{rel}}) \subseteq \D ((H^{\mathrm{rel}} + P_+ B P_+)^{\ast})$ as $P_+ B P_+$ is bounded. We compute $\D ((H^{\mathrm{rel}} + P_+ B P_+)^{\ast})$ with help of its Frobenius-Schur factorization. Recall the operators $\mathcal{R}$ and $\mathcal{T}$ from the proof of Theorem~\ref{thm_self-adjointextensionHrel}, line~\eqref{eq_def_operators_RST}: 
\begin{align}
\mathcal{R} 
=
\left(
\begin{array}{cc}
\mathrm{id} & \mathbf{0} \\
P_- \gamma V P_+ A^{- 1} \quad & \mathrm{id}
\end{array}\right) 
\, , \quad
\mathcal{T}
=
\left(
\begin{array}{cc}
\mathrm{id} & \quad \overline{A^{-1} P_+ \gamma V P_-} \\
\mathbf{0} & \mathrm{id}
\end{array}
\right) .
\end{align}
We note that one can rewrite this as $\mathcal{T}=\mathrm{id} + \overline{A^{-1} P_+ \gamma V P_-}$. We then have $\mathcal{T}^{-1}=\mathrm{id} - \overline{A^{-1} P_+ \gamma V P_-}$. In addition, we define 
\begin{align}
\mathcal{S}
:=
\left(
\begin{array}{cc}
A & \mathbf{0} \\
\mathbf{0} & S
\end{array}
\right) 
\, , \quad 
\D (\mathcal{S}) = \D_+ \oplus \D (S) .
\end{align}
This yields $(H^{\mathrm{rel}} + P_+ B P_+)^{\ast} = (\mathcal{R} \, \mathcal{S} \, \mathcal{T})^* = \mathcal{R} \, \mathcal{S}^{\ast} \, \mathcal{T}$ 
as $\mathcal{T}=\mathcal{R}^*$ and $\mathcal{R} = \mathcal{T}^*$ are bounded and boundedly invertible by Lemma~\ref{lemma_boundedmatrixop} (Appendix~\ref{section_matrixoperators}). Theorem~\ref{thm_closability} (Appendix~\ref{section_matrixoperators}) provides the domain: 
\begin{align}\label{eq_definitiondomainHrelFadjoint}
\D ((\mathcal{R} \, \mathcal{S} \, \mathcal{T})^*)
= 
\left\{ 
\left.
\left(
\begin{array}{c}
f \\
g
\end{array}
\right) 
\in 
\Hrel_+ \oplus \Hrel_- 
\right|\  
\begin{array}{c}
f + \overline{A^{-1} P_+ \gamma V P_-} g \in \D_+, \\
g \in \D (S^*) 
\end{array}
\right\} .
\end{align}
Therefore, it makes sense to again split $\D (\tilde{H}^{\mathrm{rel}})$ by the projections $P_\pm$. We thus write $f = (f_+, f_-)^\top$. 

Now, as $\mathcal{T}$ is bounded and boundedly invertible, it is a one-to-one map between 
$\D ((H^{\mathrm{rel}} + P_+ B P_+)^{\ast})$ 
and 
$\D_+ \oplus \D (S^{\ast})$.
Therefore, there exists a unique 
$\varphi = (\varphi_1, \varphi_2)^{\top} \in \D_+ \oplus \D (S^{\ast})$ 
such that 
$f =\mathcal{T}^{- 1} \varphi$, i.e., 
\begin{align}\label{eq_Trepresentationoff}
\mathcal{T}^{- 1} \varphi 
&= 
\left(
\begin{array}{cc}
\mathrm{id} & \quad - \overline{A^{- 1} P_+ \gamma V P_-}\\
\mathbf{0} & \mathrm{id}
\end{array}
\right) 
\left(
\begin{array}{c}
\varphi_1 \\
\varphi_2
\end{array}
\right) \nonumber \\
&= 
\left(
\begin{array}{c}
\varphi_1 - \overline{A^{- 1} P_+ \gamma V P_-} \varphi_2 \\
\varphi_2
\end{array}
\right) 
= 
\left(
\begin{array}{c}
f_+ \\
f_-
\end{array}
\right) .
\end{align}
Since $\varphi_1\in\D_+$ and $\varphi_2\in\D(S^*)$, this form of $f$ shows that $f$ lies in 
\begin{align}\label{eq_definitionHrelF_2}
\D (H^{\mathrm{rel}}_F)
= 
\left\{ 
\left.
\left(
\begin{array}{c}
f \\
g
\end{array}
\right) 
\in 
\Hrel_+ \oplus \Hrel_- 
\right|\  
\begin{array}{c}
f + \overline{A^{-1} P_+ \gamma V P_-} g \in \D_+, \\
g \in \D (S_F )
\end{array}
\right\} ,
\end{align}
if $f_- = \varphi_2 \in \D (S_F)$. Using $\D (S_F) = \D (V^{1/2}) \cap \D (S^*)$ by Lemma~\ref{lemma_schursym} and $f_- \in \D(S^*)$, we conclude that it suffices to show that $f_- \in \D (V^{1/2}) \cap \Hrel_-$. 

We remark that it is not immediately clear at this point that 
$\langle f, V f \rangle < \infty$ implies 
$\langle f_-, V f_- \rangle < \infty$ as, in principle, there could occur cancellations between the parts in $\Hrel_+$ and $\Hrel_-$. This, however, is not the case. 

In the following, it will turn out to be useful to define the Gaussian $\chi_n$ for all $n\in\N$ by
$\chi_n (\vec{r}) := e^{- 2^n \pi | \vec{r} |^2}$ and with it the Gaussian cut-off $1-\chi_n$ for $f$ by $f_n:=(1-\chi_n)f$. Now, by the same reasoning as in the proof of Theorem~\ref{thm_coreV1/2}, we know that $f_n \in \mathcal{D} (V)$ for all $n \in \N$, and moreover, we obtain $\|f-f_n\|\xrightarrow{n\rightarrow\infty}0$, 
$\|f_\pm - P_\pm f_n\|\xrightarrow{n\rightarrow\infty}0$, and 
$\|V^{1/2} (f-f_n)\|\xrightarrow{n\rightarrow\infty}0$ by dominated convergence. This implies
\begin{align}
\left\langle (f_m - f_n), V (f_m - f_n) \right\rangle \xrightarrow{m,n \rightarrow \infty} 0 . 
\end{align}
We express $f$ in $f_n = (1-\chi_n)f$ with $f =\mathcal{T}^{- 1} \varphi$. In order to shorten the expressions, we introduce the abbreviations
$\chi_{m,n} \equiv \chi_n - \chi_m $ and
$f_{m,n} \equiv f_m - f_n $:  
\begin{align}
f_{m,n}
&=
\chi_{m,n} \,
\mathcal{T}^{-1} \varphi \nonumber \\
&=
\chi_{m,n}
\left(
\varphi_1 - \overline{A^{-1} P_+ \gamma V P_-} f_- + f_-
\right) \nonumber \\
&=
\chi_{m,n} 
\left(
\varphi_1 + (\mathrm{id} - \overline{A^{- 1} P_+ \gamma V P_-}) f_- 
\right) \nonumber \\
&=
\chi_{m,n} 
\left(
\varphi_1 + \mathcal{T}^{-1} f_-
\right) .
\end{align}
As we have
$f_{m,n} \in \D (V)$ 
for all $m,n \in \N$, we are allowed to expand the expression 
$\langle f_{m,n}, V f_{m,n} \rangle$
as follows: 
\begin{align}
\langle f_{m,n}, V f_{m,n} \rangle
={}&
\left\langle \chi_{m,n} \left(\varphi_1 + \mathcal{T}^{-1} f_- \right), V \chi_{m,n} \left(\varphi_1 + \mathcal{T}^{-1} f_- \right) \right\rangle \nonumber \\
={}&
\left\langle 
\chi_{m,n} \varphi_1 , V \chi_{m,n} \varphi_1
\right\rangle 
+
\left\langle 
\chi_{m,n} \varphi_1 , V \chi_{m,n} \mathcal{T}^{-1} f_-
\right\rangle \nonumber \\
&{}+
\left\langle 
\chi_{m,n} \mathcal{T}^{-1} f_- , V \chi_{m,n} \varphi_1
\right\rangle \nonumber \\
&{}+
\left\langle
\chi_{m,n} \mathcal{T}^{-1} f_- , V \chi_{m,n} \mathcal{T}^{-1} f_-
\right\rangle .
\end{align}
All summands containing a $\varphi_1 \in \D_+ \subset \D (V)$ are grouped together:  
\begin{align}
G_{m,n} (\varphi_1)
:={}&
\left\langle 
\chi_{m,n} \varphi_1 , V \chi_{m,n} \varphi_1
\right\rangle 
+
\left\langle 
\chi_{m,n} \varphi_1 , V \chi_{m,n} \mathcal{T}^{-1} f_-
\right\rangle \nonumber \\
&{}+
\left\langle 
\chi_{m,n} \mathcal{T}^{-1} f_- , V \chi_{m,n} \varphi_1
\right\rangle  .
\end{align}
We first compute 
\begin{eqnarray}
\left|
\left\langle 
\chi_{m,n} \varphi_1 , V \chi_{m,n} \mathcal{T}^{-1} f_-
\right\rangle
\right| 
& = & 
\left|
\left\langle 
V \chi_{m,n} \varphi_1 , \chi_{m,n} \mathcal{T}^{-1} f_-
\right\rangle
\right| \nonumber \\
& \leq & 
\left\| 
V \chi_{m,n} \varphi_1
\right\|
\left\| 
\chi_{m,n} \mathcal{T}^{-1} f_-
\right\| 
\xrightarrow{m,n \rightarrow \infty} 0 ,
\end{eqnarray}
where both factors tend to zero as $m,n \rightarrow\infty$ by dominated convergence. Note that here it is necessary that $\varphi_1 \in \D_+ \subset \D (V)$. In an analogous way, the convergence to zero of the remaining summands of $G_{m,n} (\varphi_1)$ is proven. Hence, we get  
$\left| G_{m,n} (\varphi_1) \right| \xrightarrow{m,n \rightarrow \infty} 0$. 
This, together with 
$\langle f_{m,n}, V f_{m,n} \rangle \xrightarrow{m,n \rightarrow \infty} 0$, 
implies
\begin{align}
\left\|
V^{1/2} \chi_{m,n} \mathcal{T}^{-1} f_-
\right\|^2
&=
\left|
\left\langle
\chi_{m,n} \mathcal{T}^{-1} f_- , V \chi_{m,n} \mathcal{T}^{-1} f_-
\right\rangle 
\right|\nonumber \\
&\leq 
\langle f_{m,n}, V f_{m,n} \rangle
+\left| G_{m,n} (\varphi_1) \right|
\xrightarrow{m,n \rightarrow \infty} 0 .
\end{align}
As we also have 
\begin{align}
\left\|
(1- \chi_{m}) \mathcal{T}^{-1} f_- -
(1- \chi_{n}) \mathcal{T}^{-1} f_- 
\right\|
=
\left\|
\chi_{m,n} \mathcal{T}^{-1} f_-
\right\|
\xrightarrow{m,n\rightarrow\infty} 0
\end{align}
by dominated convergence, we can conclude that 
$((1- \chi_{n}) \mathcal{T}^{-1} f_-)_{n\in\N}$
is a Cauchy sequence in the graph norm of $V^{1/2}$. Thus, 
$\mathcal{T}^{-1} f_- \in \D (V^{1/2})$
since $V^{1/2}$ is closed. 

Since line~\eqref{eq_V12ABfinite} implies 
$\overline{A^{-1} P_+ \gamma V P_-} \; \D (V^{1/2}) \subseteq \D (V^{1/2})$, 
we also get 
\begin{align}\label{eq_domainleftinvariant}
(\mathrm{id} \pm \overline{A^{-1} P_+ \gamma V P_-}) \; \D (V^{1/2}) \subseteq \D (V^{1/2}) , 
\end{align}
i.e., $\mathcal{T}$ as well as $\mathcal{T}^{-1}$ map $\D (V^{1/2})$ into $\D (V^{1/2})$. 

Now, we assume for the moment that $f_- \not\in \D (V^{1/2})$ and aim at a contradiction. We already showed that 
$\mathcal{T}^{-1} f_- \in \mathcal{D} (V^{1/2})$,
and from line~\eqref{eq_domainleftinvariant} we know that $\mathcal{T}$ maps $\D (V^{1/2})$ into $\D (V^{1/2})$. Thus, 
\begin{align}
\mathcal{T}^{-1} f_- \in \mathcal{D} (V^{1/2})
\; \Rightarrow \;
\mathcal{T} \mathcal{T}^{-1} f_- \in \mathcal{D} (V^{1/2})
\; \Rightarrow \;
f_- \in \D (V^{1/2}) . 
\end{align}
This is a contradiction to our assumption 
$f_- \not\in \D (V^{1/2})$
and thus, 
$f_- \in \D (V^{1/2})$. 
Therefore, $f_- \in \D (V^{1/2}) \cap \D (S^{\ast}) = \D (S_F)$ 
which implies $f \in \D (H^{\mathrm{rel}}_F)$. 

For the reverse inclusion $\D (H^{\mathrm{rel}}_F) \subseteq \D (\tilde{H}^{\mathrm{rel}})$, 
it suffices to note that 
\begin{align}
\D (H^{\mathrm{rel}}_F) = \D ((H^{\mathrm{rel}}_F)^{\ast})
\subseteq \D ((\tilde{H}^{\mathrm{rel}})^{\ast}) = \D (\tilde{H}^{\mathrm{rel}}) \subseteq
\D (H^{\mathrm{rel}}_F) .
\end{align} 
Therefore, $\D (\tilde{H}^{\mathrm{rel}}) = \D (H^{\mathrm{rel}}_F)$, which concludes the proof. 
\end{proof} 

\appendix 
\section{Comment on the existing proof by Okaji et al.}\label{section_exproof} \label{section_appendix_okaji}
We briefly comment on the article {\cite{okaji_dere}} by Okaji et al. which was written in response to {\cite{derezinski_openproblems}}. As already mentioned in the introduction, the proof of their claim, i.e., essential self-adjointness of $H_{\mathrm{DC}}$ (even including external Coulomb potentials) as operator in the underlying Hilbert space 
\begin{align}
\mathcal{H}_a = \left( L^2 (\R^3) \otimes \C^4 \right) \wedge \left( L^2 (\R^3) \otimes \C^4 \right)
\end{align}
comprises a gap. In the following, we want to lay out the missing step. 

The authors employ an auxiliary operator which is denoted by $H^+$ and given by\footnote{We give $H^+$ in the same $\C^{16}$-basis as $H_{\mathrm{DC}}$ in Eq.~\eqref{eq_HDCdef}. In {\cite{okaji_dere}}, a different basis is used. } 
\begin{align}\label{eq_definitionHplus}
H^+ 
=
\mathbf{1}_4 \otimes \vec{\alpha} \cdot \left( -i \nabla_{\vec{x}} - i \nabla_{\vec{y}} \right) 
+ 
m (\beta \otimes \mathbf{1}_4 + \mathbf{1}_4 \otimes \beta) 
+ 
V 
\end{align}
where external potentials as well as the interaction potential are combined in
$V$. Its domain is taken to be 
\begin{align}
\D_a
=
\left( C_c^{\infty} (\R^3 ; \C^4) \otimes C_c^{\infty} (\R^3 ; \C^4) \right) \cap \mathcal{H}_a .
\end{align}

Now, $H_{\mathrm{DC}}$ and $H^+$ coincide as quadratic form on $\mathcal{H}_a$ but not as operator, i.e., for all 
$\varphi, \psi \in \D_a$
one has 
\begin{align}\label{eq_quadraticforms}
\left\langle \varphi, H_{\mathrm{DC}} \psi \right\rangle 
= 
\left\langle \varphi, H^+ \psi \right\rangle 
\end{align}
but
\begin{align}\label{eq_notequalasop}
H_{\mathrm{DC}} \psi \neq H^+ \psi . 
\end{align}
The former is proven in \cite{okaji_dere} in Theorem~5.4, the latter is seen as follows. Define $\psi$ by 
\begin{align}\label{eq_defpsinotindomainHplus}
\psi (\vec{x}, \vec{y})
=
\left(
\begin{array}{c}
0 \\
f (\vec{x}) \\
0 \\
f (\vec{x})
\end{array}
\right)
\otimes 
\left(
\begin{array}{c}
f (\vec{y}) \\
0 \\
f (\vec{y}) \\
0
\end{array}
\right)
-
\left(
\begin{array}{c}
f (\vec{x}) \\
0 \\
f (\vec{x}) \\
0
\end{array}
\right)
\otimes 
\left(
\begin{array}{c}
0 \\
f (\vec{y}) \\
0 \\
f (\vec{y})
\end{array}
\right)
\end{align}
where the function $f \colon \R^3 \rightarrow \C$ is chosen to be smooth and compactly supported. We see that $\psi$ is antisymmetric, and therefore, $\psi \in \D_a$. That line~\eqref{eq_notequalasop} holds, is now a straightforward calculation. 

Hence, from the essential self-adjointness of $H^+$ on $\D_a$ that was proven in
\cite{okaji_dere}, it does not follow that $\HDC$ is
essentially self-adjoint on $\D_a$. On the contrary, as the article at hand
shows, $H_0$ exhibits a non-trivial nullspace structure in the relative
coordinate, i.e., the coordinate of the interaction, whereas one always has $\Ker
(H^+) = \{ 0 \}$. 

We want to use this opportunity to add another remark on Eq.~\eqref{eq_quadraticforms} and draw attention to an interesting implication. We denote by 
\begin{align}
P_a \colon \H_2 \rightarrow \H_a
\end{align}
the orthogonal projection on $\H_a$. An interesting fact is that $P_a$ can also
have a regularizing effect which we will outline briefly in the following
without making the argument rigorous. 

For the sake of our argument, it suffices to consider the free and massless case, i.e., all potentials and masses are set to zero. We will nevertheless keep the notation $\HDC$ and $H^+$ in order to maintain the distinction between the auxiliary operator $H^+$ and the actual operator $\HDC$. Then, $\HDC$ is the operator $T$ from Section~\ref{section_relevantdefs}
\begin{align}
\HDC = T = \PMpl \otimes \mathrm{id} + \mathrm{id} \otimes \Mminp 
\end{align} 
and $H^+$ is in center-of-mass and relative coordinates of the form 
\begin{align}
H^+ =
\mathbf{1}_4 \otimes \vec{\alpha} \cdot \hat{\vec{P}} .
\end{align}

We see now that in order to have $\| H^+ f\| < \infty$ for some $f$, this $f$ must have at least $H^1 (\R^3, \dd R)$-regularity. For $\HDC$ however, this is not the case because of the nullspace structure of $\PMpl$. E.g., using the $\psi$ from line~\eqref{eq_defpsinotindomainHplus}, it is possible to construct a less regular $f$ with $\| \HDC f \| < \infty$ since $\psi \in \Ker (\PMpl)$. 

What the form equality Eq.~\eqref{eq_quadraticforms} 
$\left\langle \varphi, H_{\mathrm{DC}} \psi \right\rangle 
= 
\left\langle \varphi, H^+ \psi \right\rangle $
actually implies, is the operator equality $P_a H^+ P_a = \HDC$. However, $\HDC$ allows for less regular functions than $H^+$. Thus, we can conclude that the projection $P_a$ has a regularizing effect on $H^+ P_a f$. 

\section{Matrix operators with unbounded entries}
\label{section_matrixoperators}
In the following, we glimpse at the theory of matrix operators with unbounded entries. The goal of this appendix is to provide various relations between the unbounded entries and give the Frobenius-Schur factorization in Theorem~\ref{thm_closability}. A general reference for matrix operators is \cite{tretter}. We consider the matrix operator 
\begin{align}\label{eq_generalmatrixoperator}
\mathcal{A} = 
\left(
\begin{array}{cc}
A & B \\
C & D
\end{array}
\right) 
\end{align}
that acts in the Hilbert space 
$\mathcal{H}=\mathcal{H}_1 \oplus \mathcal{H}_2$, 
$\H_1$ and $\H_2$ being closed subspaces of $\H$, 
and where 
\begin{align}
A \colon \D(A) \rightarrow \H_1 , \quad 
B \colon \D(B) \rightarrow \H_1 , \quad
C \colon \D(C) \rightarrow \H_2 , \quad 
D \colon \D(D) \rightarrow \H_2 .
\end{align}
Throughout this appendix, we work with the following assumptions: 
\begin{enumerate}[({A}1)]
\item \label{assumptions_matrixop_1}
$A,B,C,$ and $D$ are closable, possibly unbounded operators with dense domains 
\begin{align}
\D(A), \D (C) \subset \H_1 \, , \quad 
\D(B), \D (D) \subset \H_2 . 
\end{align}

\item \label{assumptions_matrixop_2}
$\D (B) =\D (D)$. 

\item \label{assumptions_matrixop_3}
The resolvent set of $A$ is not empty, i.e. $\rho (A) \neq \emptyset$.

\item \label{assumptions_matrixop_4}
$\D (A^{\ast}) \subset \D (B^{\ast})$.

\item \label{assumptions_matrixop_5}
$\D (A) \subset \D (C)$. 
\item \label{assumptions_matrixop_6}
$\D (\mathcal{A}) =\D (A) \oplus \D (D)$ which is dense in $\H$. 
\end{enumerate}

\begin{lemma}\label{lemma_symmetryblockmatrix}
The matrix operator $\mathcal{A}$ is symmetric if and only if 
\begin{align}
A \subseteq A^{\ast} , \quad 
D \subseteq D^{\ast} , \quad 
C \upharpoonright \D (A) \subseteq B^{\ast} , \quad 
B \subseteq C^{\ast} .
\end{align}
\end{lemma}

\begin{proof}
See {\cite[Proposition 2.6.1]{tretter}}. 
\end{proof}

\begin{lemma}\label{lemma_boundedmatrixop}
The following statements hold for all $\mu \in \rho (A)$: 
\begin{enumerate}[a)]
	\item The operator $(A - \mu)^{- 1} B$ is bounded on $\D (B)$. 
	
	\item The operator $C (A - \mu)^{- 1}$ is bounded on all of $\H_1$. 
	
	\item The matrix operators 
	$\mathcal{R}(\mu) \colon \H \rightarrow \H$ and 
	$\mathcal{T}(\mu) \colon \H \rightarrow \H$, 
	given by 
	\begin{align} 
	\mathcal{R} (\mu) := 
	\left(
	\begin{array}{cc}
	\mathrm{id} & \mathbf{0} \\
	C (A - \mu)^{- 1} \quad & \mathrm{id}
	\end{array}
	\right) 
	, \quad 
	\mathcal{T} (\mu) := 
	\left(
	\begin{array}{cc}
	\mathrm{id} & \quad \overline{(A - \mu)^{- 1} B} \\
	\mathbf{0} & \mathrm{id}
	\end{array}
	\right) 
	\end{align}
	are bounded and boundedly invertible. 
	
	\item If $\mathcal{A}$ is symmetric with $A = A^*$, then $\mathcal{R}(\mu)^* = \mathcal{T}(\overline{\mu})$ and $\mathcal{T}(\mu)^* = \mathcal{R}(\overline{\mu})$ hold. 
\end{enumerate}
\end{lemma}

\begin{proof}
\begin{enumerate}[a)]
	\item See \cite[Remark~2.2.15]{tretter}. 
	
	\item This follows from the closed graph theorem. 
	
	\item See \cite[Theorem~2.2.18]{tretter}. 
	
	\item With Lemma \ref{lemma_symmetryblockmatrix} and $A = A^*$, we obtain
	\begin{align}\label{eq_1.3}
	( (A-\mu)^{-1} B )^* = B^* (A^*-\overline{\mu})^{-1} = B^* (A-\overline{\mu})^{-1} = C (A-\overline{\mu})^{-1} 
	\end{align}
	and
	\begin{align}\label{eq_1.4}
	( C (A-\mu)^{-1} )^* \supseteq (A^*-\overline{\mu})^{-1} C^* = (A-\overline{\mu})^{-1} C^* \supseteq (A-\overline{\mu})^{-1} B .
	\end{align}
	Since the bounded linear transformation theorem gives a unique closed extension of $(A-\mu)^{- 1} B$, and $(C (A-\mu)^{- 1})^{\ast}$ is closed, it follows that 
	$(C (A-\mu)^{- 1})^{\ast} = \overline{(A-\overline{\mu})^{- 1} B}$. 
	Hence, equations \eqref{eq_1.3} and \eqref{eq_1.4} imply 
	$\mathcal{R}(\mu)^* =\mathcal{T}(\overline{\mu})$ and
	$\mathcal{T}(\mu)^* =\mathcal{R}(\overline{\mu})$. 
	\qedhere
\end{enumerate}
\end{proof}

The Schur complement of $A$ is defined by
\begin{align}\label{eq_def_schurcomplement}
S (\mu) := D - \mu - C (A - \mu)^{- 1} B 
\end{align}
with domain $\D (S (\mu)) = \D (D)$ for all $\mu \in \rho (A)$. 

\begin{theorem}\label{thm_closability}
$\mathcal{A}$ is closable if and only if, for all $\mu \in \rho (A)$, $S (\mu)$ is closable in $\mathcal{H}_2$. The closure $\overline{\mathcal{A}}$ is given by the Frobenius-Schur factorization
\begin{align}\label{frob_schur_fact}
\overline{\mathcal{A}} 
= 
\mu + 
\left(
\begin{array}{cc}
\mathrm{id} & \mathbf{0} \\
C (A - \mu)^{- 1} \quad & \mathrm{id}
\end{array}
\right) 
\left(
\begin{array}{cc}
A - \mu & \mathbf{0} \\
\mathbf{0} & \overline{S (\mu)}
\end{array}
\right) 
\left(
\begin{array}{cc}
\mathrm{id} & \quad \overline{(A - \mu)^{- 1} B} \\
\mathbf{0} & \mathrm{id}
\end{array}
\right) ,
\end{align}
independently of $\mu \in \rho (A)$, that is,
\begin{align} 
\D (\overline{\mathcal{A}}) 
= 
\left\{ 
\left.
\left(
\begin{array}{c}
f\\
g
\end{array}
\right) 
\in \mathcal{H}_1 \oplus \mathcal{H}_2 
\right|\
\begin{array}{c}
f + \overline{(A - \mu)^{- 1} B} g \in \D (A), \\
g \in \D (\overline{S (\mu)}) 
\end{array}
\right\} 
\end{align}
\begin{align} \overline{\mathcal{A}}  \left(\begin{array}{c}
f\\
g
\end{array}\right) = \left(\begin{array}{c}
(A - \mu) (f + \overline{(A - \mu)^{- 1} B} g) + \mu f\\
C (f + \overline{(A - \mu)^{- 1} B} g) + (\overline{S (\mu)} + \mu) g
\end{array}\right) . 
\end{align}
\end{theorem}

\begin{proof}
See {\cite[Theorem 1]{shkalikov95}}. 
\end{proof}

\section{Auxiliary lemma} \label{section_auxlemmas}

\begin{lemma} \label{lemma_modtriangle}
Let $0 < \kappa \leq 1$. Then, for any $\vec{a}, \vec{b} \in \R^3$ it holds that 
\begin{align}
\left| \left| \vec{a}-\vec{b} \right|^{\kappa / 2} - \left| \vec{a} \right|^{\kappa / 2} \right| 
\leq 
\left| \vec{b} \right|^{\kappa / 2} .
\end{align}
\end{lemma}

\begin{proof}
First, we assume $| \vec{a}-\vec{b} |^{\kappa / 2} \geq | \vec{a} |^{\kappa / 2}$. Then, with $| \vec{a}-\vec{b} | \leq | \vec{a} | + | \vec{b} |$ and monotonicity of exponentiation, 
we obtain
\begin{align}\label{eq_modtriangle_1}
| \vec{a}-\vec{b} |^{\kappa / 2} 
\leq 
\left(
| \vec{a} | + | \vec{b} |
\right)^{\kappa / 2}
\leq 
| \vec{a} |^{\kappa/2} + | \vec{b} |^{\kappa/ 2} ,
\end{align}
where we used the equivalence of the $2/\kappa$-norm with the $1$-norm in $\R^2$: 
\begin{align}\label{eq_equivalencewithonenorm}
\left(
| \vec{a} | + | \vec{b} |
\right)^{\frac{1}{2/\kappa}}
& = 
\left(
| \vec{a} |^{\frac{2/\kappa}{2/\kappa}} + | \vec{b} |^{\frac{2/\kappa}{2/\kappa}}
\right)^{\frac{1}{2/\kappa}} \nonumber \\
& =
\left\|
\left(
\begin{array}{c}
| \vec{a} |^{\frac{1}{2/\kappa}} \\
| \vec{b} |^{\frac{1}{2/\kappa}}
\end{array}
\right)
\right\|_{2/\kappa}  
\leq 
\left\|
\left(
\begin{array}{c}
| \vec{a} |^{\frac{1}{2/\kappa}} \\
| \vec{b} |^{\frac{1}{2/\kappa}}
\end{array}
\right)
\right\|_{1} 
=
| \vec{a} |^{\frac{1}{2/\kappa}} + | \vec{b} |^{\frac{1}{2/\kappa}} .
\end{align}
Inequality~\eqref{eq_modtriangle_1} then implies 
\begin{align}
| \vec{a}-\vec{b} |^{\kappa / 2} - | \vec{a} |^{\kappa / 2} \leq
| \vec{b} |^{\kappa / 2} . 
\end{align}

Next, we assume $| \vec{a} |^{\kappa / 2} \geq | \vec{a}-\vec{b}|^{\kappa / 2}$. Then, with $| \vec{a} | \leq | \vec{a}-\vec{b}    | + | \vec{b} |$ and monotonicity of exponentiation, we get
\begin{align}
| \vec{a} |^{\kappa / 2}
\leq 
\left(
| \vec{a} - \vec{b} | + | \vec{b} |
\right)^{\kappa / 2} 
\leq 
| \vec{a}-\vec{b} |^{\kappa / 2} + | \vec{b} |^{\kappa / 2}
\end{align}
which implies
\begin{align}
| \vec{a} |^{\kappa / 2} - | \vec{a}-\vec{b} |^{\kappa / 2} \leq
| \vec{b} |^{\kappa / 2} 
\end{align}
where we made again use of the inequality from line~\eqref{eq_equivalencewithonenorm}. This concludes the proof. 
\end{proof}

\subsection*{Acknowledgment}
We gratefully acknowledge fruitful discussions with Detlef Dürr and Martin Kolb. 

\bibliographystyle{plain}
\bibliography{references}

\end{document}